\documentclass{amsart}
\usepackage{amssymb}

\usepackage{amsfonts}
\usepackage{geometry}

\newtheorem{theorem}{Theorem}
\theoremstyle{plain}

\newtheorem{claim}{Claim}

\newtheorem{corollary}{Corollary}

\newtheorem{definition}{Definition}

\newtheorem{proposition}{Proposition}
\newtheorem{remark}{Remark}

\numberwithin{equation}{section}
\input{tcilatex}

\begin{document}
\title[Agenda manipulation-proofness]{ Agenda manipulation-proofness,
stalemates, and redundant elicitation in preference aggregation. \\
{\tiny Exposing the bright side of Arrow's theorem. }}
\author{Stefano Vannucci}
\address{University of Siena, Department of Economics and Statistics}
\email{stefano.vannucci@unisi.it}
\maketitle

\begin{abstract}
This paper provides a general framework to explore the possibility of agenda
manipulation-proof and proper consensus-based preference aggregation rules,
so powerfully called in doubt by a disputable if widely shared understanding
of Arrow's `general possibility theorem'. We consider two alternative
versions of agenda manipulation-proofness for social welfare functions, that
are distinguished by `parallel' vs. `sequential' execution of agenda
formation and preference elicitation, respectively. Under the `parallel'
version, it is shown that a large class of anonymous and idempotent social
welfare functions that satisfy both agenda manipulation-proofness and
strategy-proofness on a natural domain of single-peaked `meta-preferences'
induced by arbitrary total preference preorders are indeed available. It is
only under the second, `sequential' version that agenda
manipulation-proofness on the same natural domain of single-peaked
`meta-preferences' is in fact shown to be tightly related to the classic
Arrowian `independence of irrelevant alternatives' (IIA) for social welfare
functions. In particular, it is shown that using IIA to secure such
`sequential' version of agenda manipulation-proofness and combining it with
a very minimal requirement\ of distributed responsiveness results in a
characterization of the `global stalemate' social welfare function, the
constant function which invariably selects universal social indifference. It
is also argued that, altogether, the foregoing results provide new
significant insights concerning the actual content and the \textit{%
constructive} implications of Arrow's `general possibility theorem' from a
mechanism-design perspective.


JEL Classification D71


Author's note: The author has no conflict of interest to acknowledge
concerning the present work.
\end{abstract}

\section{Introduction}

The \textit{agenda} is a key part of any decision problem, and is specified
by its two components. The first component of the agenda, its \textit{content%
}, is defined by the alternative options it contains. The second component,
its \textit{structure}, is the rule governing the process of option-scrutiny
that is required in order to solve the given decision problem. Thus, such a
structure amounts to a \textit{total preorder} of the available options,
resulting in an ordered partition of options that channels their \textit{%
sequential} scrutiny. It follows that \textit{agenda control }is perforce a
crucial issue, and has at least \textit{two }dimensions, namely
agenda-content and agenda-structure control. Therefore \textit{agenda
manipulation, }or the exercise of agenda control to influence the final
decision, is also a `two-dimensional' activity which must be taken into
account, to be possibly prevented or restrained.

As it happens, concerns for \textit{agenda manipulation along both of its
dimensions} have always played a distinguished role in the literature on
collective decision-making, but have scarcely if ever been the target of
explicit treatment within formal models of preference aggregation. To be
sure, it is well-known that under transitivity of the relevant preference
relation, preference-maximizing choices are \textit{path-independent, }%
namely do not depend on the sequence of intermediate choices and rejections
out of the sequence of subsets which is dictated by the agenda-structure.
Hence, whenever the aggregation rule is a \textit{social welfare function},
meaning that both the individual preference relations to be aggregated and
the `social' or aggregate preference are \textit{transitive} (and \textit{%
total}), agenda-structure control is virtually inconsequential, and agenda
manipulation along the agenda-structure dimension is automatically
prevented. On the contrary, agenda-content manipulation is of course still
possible even if all the relevant preference relations are transitive. Since
we are going to focus precisely on social welfare functions, in the sequel
we shall only discuss agenda-content manipulation, disregarding entirely
agenda-structure manipulation. Therefore, in the rest of the present paper
we shall simply \textit{identify agendas with their contents}\footnote{%
By contrast, most contributions to agenda control in the political science
literature are focussed on \textit{agenda-structure} manipulation: see e.g.
Schwartz (1986), part IV, and Austen-Smith, Banks (2005) for an extensive
treatment of agenda-structure control models in political science. Miller
(1995) explicitly distinguishes the two dimensions of agenda control but
then provides a review of some models offering a \textit{joint} treatment of
them.}. Accordingly, `agenda manipulation' is henceforth used, for the sake
of simplicity, as a synonym for `agenda-content manipulation'.

Indeed, in his classic \textit{Social Choice and Individual Values} (1963)
Arrow used precisely \textit{the need to prevent agenda manipulation} as the
main argument in favor of \ his own `Independence of Irrelevant
Alternatives' (IIA), as a key condition for \textit{proper} (or non-trivial%
\footnote{%
A dictatorial social welfare function is of course unanimity-respecting and
thus in a sense also consensus-based, but only trivially so.})\textit{\
consensus-based social welfare functions,} which are the main focus of that
work. Since then, the notion that IIA should be regarded as a basic
`nonmanipulability property' has been further reinforced by the rise and
enormous proliferation of models focussing on \textit{strategic manipulation}
issues in preference aggregation, to become eventually almost common place.
Yet, within standard models of preference aggregation including social
welfare functions it is \textit{just preference relations} that agents
provide as \textit{inputs} while the relevant agenda is a \textit{parameter}
of the aggregation rule. But then, such an aggregation rule is nothing else
than the relevant strategic \textit{game-form. }It follows that
agenda-manipulation amounts to a \textit{structural manipulation }of the
very `aggregation game', literally a game-changer. Hence, the exact
connection between agenda manipulation-proofness and IIA is \textit{not}
amenable to a proper game-theoretic scrutiny \textit{unless} the preference
aggregation model is expanded to involve the \textit{agenda formation process%
} itself. In particular, such an expanded model is needed to establish
whether the full force of IIA is actually necessary to prevent agenda
manipulation for social welfare functions that are at least minimally
outcome-unbiased and agent-inclusive.

Thus, some explicit formulation of \ the agenda formation process has to be
introduced in the relevant preference aggregation model. In the present work 
\textit{two main types of agenda formation protocols} are considered. Both
of them rely on a prespecified \textit{admissible set of }outcomes out of
which the \textit{actual} \textit{agenda} has to be defined. Moreover, in
order to avoid any sort of infinite regress, we can safely assume that
outcome-admissibility is established by another (possibly `democratic', but 
\textit{distinct}) procedure\footnote{%
It is worth recalling here that Dahl (1956) famously suggested to label as
`populist democracy' the doctrine that advocates reliance on the simple
majority rule to settle every issue including the identification of the 
\textit{admissible issues }for a possible public agenda. Arguably, one might
invoke a generalized notion of `populist democracy' as the advocacy of a
unique `democratic' decision rule to settle every issue, including every
aspect pertaining to agenda control. In that connection, assuming that
admissibility of outcome sets is subject to a \textit{distinct protocol }(if
possibly also `democratic' in some appropriate sense) also amounts to
preventing any `populist' interpretation of the overall decision mechanism.}%
. In the first agenda formation protocol, however, agents provide \textit{at
once} both their preferences on admissible outcomes and their proposals
concerning the agenda. In the second one, on the contrary, a \textit{first} 
\textit{stage} is devoted to \textit{specifying the actual agenda, }and%
\textit{\ }is followed by a \textit{second, preference-elicitation} stage
where the agents express their preferences on the previously chosen agenda.
Accordingly, two distinct formulations of \textit{Agenda
Manipulation-Proofness} (AMP) are introduced, and their distinctive impact
on the design of preference aggregation rules that guarantee at least a
minimal amount of \textit{outcome-unbiased distributed responsiveness }to
individual preferences is explored and discussed at length. In particular,
it is shown that under the \textit{first }formulation\textit{, }AMP \textit{%
and }Strategy-Proofness on a comprehensive single-peaked domain of
`meta-preferences' most naturally induced by basic preferences on outcomes
are shared by a large class of preference aggregation rules on the full
domain of \textit{arbitrary} profiles of total preference preorders\footnote{%
As mentioned below, such a result relies heavily on the main theorem of
Savaglio, Vannucci (2021) concerning strategy-proof aggregation rules in
median join-semilattices.}. Such a class of strategy-proof aggregation rules
includes social welfare functions which indeed satisfy a remarkable
combination of valuable properties. Namely, anonymity, monotonicity and a 
\textit{basic} version of Pareto-optimality, possibly even (weak)
neutrality, though occasionally producing \textit{stalemates\footnote{%
A stalemate is defined as `social indifference' among a set of alternative
social states including a pair $x,y$ such that $x$ is unanimously strictly
preferred to $y$. Thus, by definition, a stalemate admits of violations of
the Weak Pareto principle (which enforces strict social preference for $x$
versus $y$ under the aforementioned situation).}} as an output to some
preference profiles exhibiting certain \textit{specific} patterns of strong
conflict (e.g. Condorcet cycles)\footnote{%
The above mentioned social welfare functions are the quota rules, including
the Condorcet-Kemeny median rule which is indeed neutral when the number of
agents is odd.}. By contrast, the \textit{second version} of AMP turns out
to be strictly related to IIA, and the combination of IIA with a couple of
much weaker, indeed minimal, requirements of unbiased and distributed
responsiveness provides a characterization of the \textit{Global Stalemate
constant social welfare function}, namely the social welfare functions which
has \textit{universal indifference }as its\textit{\ unique possible output.}%
\footnote{%
Such a result, which amounts to a considerable strenghtening of a previous
characterization of the same constant social welfare function due to Hansson
(1969), relies heavily on Wilson (1972) and Savaglio,Vannucci (2021).}
Moreover, if the Weak Pareto property or even just idempotence (namely,
`respect for unanimity') is adjoined to IIA, an Arrowian impossibility
result is obtained. \ Thus, in order to secure agenda manipulation-proofness
of a social welfare function one may consider \textit{two} basic alternative
approaches having strikingly \textit{different} consequences. One of those
approaches relies on the introduction of IIA: it was correctly identified by
Arrow's seminal contribution, and paves the way to his classic
characterization of \textit{dictatorial} social welfare functions. That
result signals an important obstruction to the design of social welfare
functions as democratic preference aggregation protocols. The other
approach, however, has \textit{no connection whatsoever to IIA}, and is
consistent with a large class of anonymous, unanimity-respecting social
welfare functions that also retain a basic version of Pareto optimality
involving nonstrict preferences. We argue that the very contrast between
those two approaches and their respective results makes it possible to
single out and appreciate the\  \textit{constructive} implications of
Arrowian `impossibility theorems' concerning the design of preference
aggregation rules, as a significant part of their actual meaning and content.

The rest of this paper is organized as follows: section 2 collects the
formal description of the model and the results; section 3 consists in a
brief discussion of a few most strictly related contributions (the
interested reader is addressed to the supplementary Appendix for a more
extensive and detailed discussion of a carefully selected sample of the
massive amount of related literature);\ section 4 provides some concluding
remarks, and prospects for future research.

\section{Model and results}

\textbf{1. Preliminaries.}

Let $A$ be a nonempty finite set of alternative social states with $|A|\geq3 
$, $\mathcal{R}_{A}$ the set of all total preorders (i.e. reflexive,
transitive and connected binary relations) on $A$, $\mathcal{L}_{A}\subseteq%
\mathcal{R}_{A}$ the set of all linear orders (or \textit{antisymmetric}
total preorders on $A$), and $\mathcal{P}(A)$ the set of parts of $A$, or
possible \textit{agendas }from $A$. Let $N=\left \{ 1,..,n\right \} $ denote
a \textit{finite} population of agents/voters. We assume that $n\geq3$ in
order to avoid tedious qualifications. The subsets of $N$ are also referred
to as \textit{coalitions, }and $(\mathcal{P}(N),\subseteq)$ denotes the
partially ordered set of coalitions induced by set-inclusion. An \textit{%
order filter} of $(\mathcal{P}(N),\subseteq)$ is a set $F\subseteq \mathcal{P%
}(N)$ of coalitions such that for any $S\in F$ and any $T\subseteq N$, if $%
S\subseteq T$ then $T\in F$. The \textit{basis }of order filter $F$ is the
set of inclusion-minimal elements/coalitions of $F$, and is denoted by $%
F^{\min}$.

Each agent $i\in N$ is endowed with a total preference preorder $R_{i}\in$ $%
\mathcal{R}_{A}$ (whose \textit{asymmetric }component or \textit{strict
preference }is denoted by $P(R_{i})$), and proposes an agenda $%
A_{i}\subseteq A$. A \textit{social welfare function} for $(N,A)$ if a
function $f:\mathcal{R}_{A}^{N}\rightarrow \mathcal{R}_{A}$. We shall also
consider two types of social welfare functions enriched with an \textit{%
endogenous agenda formation }process. According to the class of \textit{%
parallel }rules \textit{agents release concurrently their entire inputs}
consisting of a total preorder on the set of \textit{all }admissible
alternatives and of a proposed agenda: preference aggregation and agenda
formation are also mutually concurrent processes. Notice that this also
entails that typically, namely whenever the selected agenda is a \textit{%
proper} subset of $A$, the elicited individual preferences turn out to be 
\textit{redundant.}

Thus, a \textit{parallel agenda-formation-enriched (PAFE) social welfare
function }for $(N,A)$ is an aggregation rule $\mathbf{f}:(\mathcal{P}%
(A)\times \mathcal{R}_{A})^{N}\rightarrow \mathcal{P}(A)\times \mathcal{R}%
_{A}\mathcal{\ }$(with projections $\mathbf{f}_{1}$ and $\mathbf{f}_{2}$ on $%
\mathcal{P}(A)$ and $\mathcal{R}_{A}\mathcal{\ }$, respectively). In
particular, such a PAFE $\mathbf{f}$ is said to be \textit{decomposable }if
and only if it can be decomposed into two component aggregation rules: an
agenda formation rule $f^{(1)}:\mathcal{P}(A)^{N}\longrightarrow \mathcal{P}%
(A)$ and a social welfare function $f^{(2)}:\mathcal{R}_{A}^{N}\rightarrow%
\mathcal{R}_{A}$ \footnote{%
In other terms, the two projections of $\mathbf{f}$, namely $\mathbf{f}_{1}:(%
\mathcal{P}(A)\times \mathcal{R}_{A}^{T})^{N}\longrightarrow \mathcal{P}(A)$
and
\par
$\mathbf{f}_{2}:(\mathcal{P}(A)\times \mathcal{R}_{A}^{T})^{N}%
\longrightarrow \mathcal{R}_{A}^{T}$ are such that for every $%
R_{N},R_{N}^{\prime}\in(\mathcal{R}_{A}^{T})^{N}$ and $B_{N},B_{N}^{\prime}%
\in(\mathcal{P}(A))^{N}:$%
\par
$\mathbf{f}_{1}(B_{N},R_{N})=\mathbf{f}_{1}(B_{N},R_{N}^{%
\prime}):=f^{(1)}(B_{N})$%
\par
and
\par
$\mathbf{f}_{2}(B_{N},R_{N})=\mathbf{f}_{2}(B_{N}^{%
\prime},R_{N}):=f^{(2)}(R_{N})$.
\par
Strictly speaking $\mathbf{f}$ and $f^{(1)}\times f^{(2)}$are isomorphic
functions, hence we also use the notation $\mathbf{f}\simeq f^{(1)}\times
f^{(2)}$ to denote that fact.}. Conversely, for any $f^{\prime}$ be an
agenda formation rule $f^{\prime}$ for $(N,A)$ and any social welfare
function $f$ for $(N,A)$ a decomposable PAFE social welfare function $%
\mathbf{f}\simeq f^{\prime}\times f$ can be defined in an obvious way.
Hence, \textit{any} social welfare function can be regarded as a component
of a decomposable PAFE social welfare function by combining it with an
agenda formation rule.

Observe that a decomposable PAFE $\mathbf{f}$ also induces a family of
functions $\mathcal{F}_{\mathbf{f}}:=\left \{ f_{B}^{(2)}:\mathcal{R}%
_{A}^{N}\rightarrow \mathcal{R}_{B}\text{ }|\text{ }B\in f^{(1)}[\mathcal{P}%
(A)]\right \} $ where $f_{B}^{(2)}(R_{N}):=(f^{(2)}(R_{N}))_{|B}$ for any $%
R_{N}\in \mathcal{R}_{A}^{N}$. Accordingly, the possible values of functions
in $\mathcal{F}_{\mathbf{f}}$ are given by a family of total preorders,
namely $\left \{ f_{B}^{(2)}(R_{N})\right \} _{R_{N}\in \mathcal{R}%
_{A}^{N},B\in \mathcal{P}(A)]}$: thus, for every $R_{N}\in \mathcal{R}_{A}^{N}
$ and $B\in \mathcal{P}(A)$, $f_{B}^{(2)}(R_{N})\in \mathcal{R}_{B}\subseteq
\dbigcup \limits_{B\in \mathcal{P}(A)}\mathcal{R}_{B}$. In particular, we
shall focus on \textit{sovereign }agenda-formation rules $f^{(1)}:\mathcal{P}%
(A)^{N}\longrightarrow \mathcal{P}(A)$ (namely, such that for every $%
C\subseteq A$ there exists $B_{N}\in \mathcal{P}(A)^{N}$ with $%
f^{(1)}(B_{N})=C$).

Under the class of \textit{sequential }rules, on the contrary, \textit{%
agents release their inputs in two steps}: first they provide concurrently
their proposed agendas to be aggregated into a shared agenda, then they
submit concurrently their preferences on the previously determined actual
agenda as their input to preference aggregation itself. Notice that in this
case, no redundancy in preference elicitation is to be expected.\footnote{%
This is in fact a most important feature that distinguishes PAFE and SAFE
social welfare functions. In principle, one could also consider a
preference-first version of SAFE social welfare functions. However, that
version would require a lot of redundance in preference elicitation, either
by eliciting preferences on the entire set $A$ (or even, most impractically,
eliciting specific preferences on every relevant subset of $A$). Now, the
first and more practical option gives rise to an aggregation procedure that
is essentially equivalent to a PAFE social welfare function (since
preferences on $A$ would be adapted to subsets by restriction). That is why
in the present work SAFE social welfare functions are actually identified
with their agenda-first variety.} Thus, a \textit{sequential
agenda-formation-enriched (SAFE) social welfare function }for $(N,A)$ is in
fact an \textit{agenda-contingent social welfare function, }namely a pair%
\textit{\ }$\widehat{\mathbf{f}}=(\widehat{f}^{1},\mathcal{F}(\widehat{f}%
^{1}))$ consisting of an agenda formation rule $\widehat{f}^{1}:\mathcal{P}%
(A)^{N}\longrightarrow \mathcal{P}(A)$, and a family

$\mathcal{F}(\widehat{f}^{1})=\left \{ \widehat{f}_{B}:\mathcal{R}%
_{B}^{N}\rightarrow \mathcal{R}_{B}\right \} _{B\in \widehat{f}^{1}[\mathcal{%
P}(A)]}$of \ possible social welfare functions, one for each possible agenda
selected by $\widehat{f}^{1}$. A particular case of special interest obtains
when the family of agenda-contingent social welfare functions is the \textit{%
uniform family }induced by its $\widehat{f}_{A}$, namely $\mathcal{F}(%
\widehat{f}^{1}):=\left \{ \widehat{f}_{B}:\mathcal{R}_{B}^{N}\rightarrow%
\mathcal{R}_{B}|\text{ }\widehat{f}_{B}:=(\widehat{f}_{A})_{|B}\right \}
_{B\in \widehat{f}^{1}[\mathcal{P}(A)]}$. In any case, again, the values of
possible social welfare functions according to $\widehat {\mathbf{f}}$ are
in $\dbigcup \limits_{B\in \mathcal{P}(A)}\mathcal{R}_{B}$.

As a consequence, \ SAFE and (decomposable) PAFE social welfare functions
essentially share $\dbigcup \limits_{B\in \mathcal{P}(A)}\mathcal{R}_{B}$ as
their common outcome space.

Therefore, it transpires that SAFE and PAFE social welfare functions for a
pair $(N,A)$ consist of functions whose domains and codomains result from
various combinations of `building blocks' chosen from the collection given
by $\mathcal{P}(A)$ and the sets $\mathcal{R}_{B}$ of all total preorders on 
$B$, for any $B\subseteq A$. As it turns out, such `building blocks' share a
common structure: all of them are \textit{median join-semilattices, }and
that fact will play a key role in the subsequent analysis of the behaviour
of PAFE and SAFE social welfare functions. Accordingly, we turn now to
providing a precise definition of median join-semilattices, and establishing
the previous claim on $\mathcal{P}(A)$ and the sets of the family $\left \{ 
\mathcal{R}_{B}\right \} _{B\subseteq A}$.

\begin{definition}
A (finite)\textbf{\ join-semilattice }is a pair $\mathcal{X}=(X,\leqslant)$ 
\textit{\ where }$X$ is a (finite) set and $\leqslant$ is a partial order
(i.e. a reflexive, transitive and antisymmetric binary relation) such that
the least upper bound or join $x\vee y$ (with respect to $\leqslant$) is
well-defined in $X$ for all $x,y\in X$ \ and thus $\vee:X\times X\rightarrow
X$ is a well-defined associative and \ commutative function that also
satisfies idempotency, namely $x\vee x=x$ for every $x\in X$.
\end{definition}

\begin{remark}
Thus a join-semilattice $\mathcal{X}=(X,\leqslant)$ can also be regarded as
a pair $(X,\vee)$ where $\vee:X\times X\longrightarrow X$ is an associative,
commutative and idempotent operation such that, for any $x,y\in X$, $x\vee
y=x$ iff $y\leqslant x$. Note that a \textit{partial }meet\textit{-}%
operation $\wedge:X\times X\longrightarrow X$ is also definable in $\mathcal{%
X}$ by means of the following rule: for any $x,y\in X$, $x\wedge y$ is the
(necessarily unique, whenever it exists) $z\in X$ such that: (i) $x\vee z=x$%
, $y\vee z=y$, and (ii) $v\vee z=z$ \ for every $v\in X$ which satisfies (i).
\end{remark}

Observe that a finite join-semilattice $\mathcal{X}=(X,\leqslant)$ has a
(unique) universal upper bound or \textit{top element} $\mathbf{1}=\vee
X=\wedge \varnothing$, and its \textit{co-atoms} are those elements $x\in X$
such that $x\ll \mathbf{1}$ (i.e. $x<\mathbf{1}$ and there is no $z\in X$
such that $x<z<\mathbf{1}$): the set of co-atoms of $\mathcal{X}%
=(X,\leqslant)$ is denoted by $\mathcal{C}_{\mathcal{X}}$ . An element $x\in
X$ is \textit{meet-irreducible }if for any $Y\subseteq X$, $x=\wedge Y$
entails $x\in Y$. \ Moreover, for any $Y\subseteq X$, $\vee Y$, respectively
is well-defined if and only if there exists $z\in X$ such that $y\leqslant z$
for all $y\in Y$, namely the elements of $Y$ have a \textit{common upper
bound}. The set of all meet-irreducible elements of $\mathcal{X}%
=(X,\leqslant)$ will be denoted by $M_{\mathcal{X}}$. Notice that, by
construction, for every $x\in X$, $x=\wedge M(x)$ where $M(x):=\left \{ m\in
M_{\mathcal{X}}:x\leqslant m\right \} $. By construction, a co-atom is also
a meet-irreducible element, but the converse need not be true. When co-atoms
and meet-irreducibles do in fact \textit{coincide} the join-semilattice is
said to be \textit{coatomistic}.\footnote{%
Dually, a (finite)\textbf{\ meet-semilattice }is a pair $\mathcal{X}%
=(X,\leqslant)$ \textit{\ }where\textit{\ }$X$ is a (finite) set, and
partial order $\leqslant$ is such that the greatest lower bound or \textit{%
meet} $x\wedge y$ (with respect to $\leqslant$) is well-defined in $X$ for
all $x,y\in X$ \ and thus $\wedge:X\times X\rightarrow X$ is a well-defined
associative and \ commutative function that also satisfies idempotency,
namely $x\wedge x=x$ for every $x\in X$.
\par
An element $x\in X$ of a meet-semilattice is\textit{\ join-irreducible }if
for any $x=\vee Y$ entails $x\in Y$ for any (finite) $Y\subseteq X$ such
that $\vee Y$, is well-defined. The set of all join-irreducible elements of $%
\mathcal{X}=(X,\leqslant)$ is denoted by $J_{\mathcal{X}}$. The \textit{atoms%
} of $\mathcal{X}$ are those elements $x\in X$ such that $\mathbf{0}\ll x$
(i.e. $\mathbf{0}<x$ and there is no $z\in X$ such that $\mathbf{0}<x<z$
where $\mathbf{0}=\wedge X=\vee \varnothing$): the set of atoms of $\mathcal{%
X}$ is denoted by $\mathcal{A}_{\mathcal{X}}$. Clearly, $\mathcal{A}_{%
\mathcal{X}}\subseteq J_{\mathcal{X}}$. The semilattice $\mathcal{X}$ \ is 
\textit{atomistic }if $\mathcal{A}_{\mathcal{X}}=J_{\mathcal{X}}$.
\par
Notice that, by construction, for every $x\in X$, $x=\vee J(x)$ where $%
J(x):=\left \{ j\in J_{\mathcal{X}}:j\leqslant x\right \} $.
\par
If a join-semilattice $\mathcal{X}=(X,\leqslant)$ is \textit{also }a
meet-semilattice then $\mathcal{X}$ is a \textbf{lattice }and \textit{%
absorption laws }hold, namely for any $x,y\in X$, $x\vee(y\wedge x)=x$ $%
=x\wedge(y\vee x)$.}

\begin{definition}
(\textbf{Median join-semilattice)\ }A (finite) \textit{join-semilattice} $%
\mathcal{X}=(X,\leqslant)$ is a \textbf{median}\textit{\ }\textbf{%
join-semilattice }if it also satisfies the following pair of conditions:

i) \textbf{upper distributivity}\textit{:} for all $u\in X$, and for all $%
x,y,z\in X$ such that $u$ is a lower bound of $\left \{ x,y,z\right \} $, $%
x\vee(y\wedge z)=(x\vee y)\wedge(x\vee z)$ (or, equivalently, $x\wedge(y\vee
z)=(x\wedge y)\vee(x\wedge z)$) holds i.e. $(\uparrow u,\leqslant_{|\uparrow
u})$ -where $\leqslant_{|\uparrow x}$denotes the restriction of $\leqslant$
to $\uparrow u$- is a \textit{distributive lattice\footnote{%
A partially ordered set $(Y,\leqslant)$ is a \textit{distributive lattice}
iff, for any $x,y,z\in X$, $x\wedge y$ and $x\vee y$ exist, and $%
x\wedge(y\vee z)=(x\wedge y)\vee(x\wedge z)$ (or, equivalently, $%
x\vee(y\wedge z)=(x\vee y)\wedge(x\vee z)$). Moreover, a (distributive)
lattice $\mathcal{X}$ is said to be \textit{lower (upper) bounded }if there
exists $\bot \in X$ \ ($\top \in X)$ such that $\bot \leqslant x$ $\ $($%
x\leqslant \top$) for all $x\in X$, and \textit{bounded}\textbf{\ }if it is
both lower bounded and upper bounded. A bounded distributive lattice $%
(X,\leqslant)$ is \textit{Boolean}\textbf{\ }if for each $x\in X$ there
exists a \textit{complement }namely an $x^{\prime}\in X$ such that $x\vee
x^{\prime}=\top$ and $x\wedge x^{\prime}=\bot$.}};

(ii) \textbf{co-coronation}\textit{\ (or meet-Helly property): }for all $%
x,y,z\in X$ if $x\wedge y$, $y\wedge z$ and $x\wedge z$ exist, then $%
(x\wedge y\wedge z)$ also exists\textit{.}
\end{definition}

In fact, it is easily checked that if $\mathcal{X}=(X,\leqslant)$ is a
median join-semilattice then the partial function $\mu_{\mathcal{X}%
}:X^{3}\rightarrow X$ defined as follows: for all $x,y,z\in X$, $\  \mu_{%
\mathcal{X}}(x,y,z)=(x\vee y)\wedge(y\vee z)\wedge(x\vee z)$

is in fact a \textit{well-defined ternary operation }on $X$\textit{, }the%
\textit{\ }\textbf{median }of $\mathcal{X}$ which satisfies the following
two characteristic properties (see Sholander (1952, 1954)):

$(\mu_{1})$ \ $\mu_{\mathcal{X}}(x,x,y)=x$ for all $x,y\in X$

$(\mu_{2})$ $\mu_{\mathcal{X}}(\mu_{\mathcal{X}}(x,y,v),\mu_{\mathcal{X}%
}(x,y,w),z)=\mu_{\mathcal{X}}(\mu_{\mathcal{X}}(v,w,z),x,y)$

for all $x,y,v,w,z\in X$.

A pair $(X,\mu)$ where $\mu$ is a ternary operation on $X$ that satisfies $%
(\mu_{1})$ and $(\mu_{2})$ is also said to be a \textit{median algebra}.

Relying on $\mu_{\mathcal{X}}$, a ternary (median-induced) \textbf{%
betweenness }relation

$B_{\mathcal{\mu}_{\mathcal{X}}}:=\left \{ (x,z,y)\in X^{3}:z=\mu _{\mathcal{%
X}}(x,y,z)\right \} $ can also be defined on $X$. \footnote{%
It should be recalled that such a median operation $\mu$ is also
well-defined in any \textit{distributive lattice }$\mathcal{X}=(X,\leqslant) 
$. Thus, every (finite) distributive lattice is in particular a (finite)
median join-semilattice (and a (finite) median meet-semilattice as well).}
The pair $(X,B_{\mathcal{\mu}_{\mathcal{X}}})$ is also said to be a \textit{%
median (ternary) space.}

\bigskip

\begin{remark}
It is worth emphasizing here that any finite median join-semilattice is
naturally endowed with two equivalent \textbf{metrics}\footnote{%
In a (finite) median join-semilattice $\mathcal{X}=(X,\leqslant)$ a \textit{%
metric }$d_{r}:X\times X\rightarrow \mathbb{Z}_{+}$ can be defined in a
natural way by the following rule: for any $x,y\in X$, $d_{r}(x,y)=2r(x\vee
y)-r(x)-r(y)$, where r is a \textit{rank function} of $\mathcal{X}$, namely
a function $r:X\rightarrow \mathbb{Z}_{+}$ such that, for any $x,y\in X$, $%
r(y)=r(x)+1$ whenever $x$ is an immediate $\leqslant$-predecessor of $y$.
\par
This metric turns out to be equivalent to the metric $\delta_{C(\mathcal{X}%
)} $ induced on $\mathcal{X}$ by the length of shortest path between any two
elements on the graph defined by the Hasse diagram of $\mathcal{X}$ (the
simple undirected graph having $X$ as its set of vertices, with edges
connecting each pair consisting of a vertex and one of its immediate $%
\leqslant$-predecessors).}, and that a further betweenness relation can be
defined on it relying on such metrics. However, it turns out that such a
metric-based betweenness is in fact equivalent to the median-based
betweenness $B_{\mu_{\mathcal{X}}}$ introduced above in the text (see e.g.
Sholander (1954), Avann (1961)).
\end{remark}

\bigskip

Thus, $B_{\mu_{\mathcal{X}}}$ is indeed a most natural `intrinsic'
betweenness relation and can also be regarded as `the' natural metric
betweenness attached to $\mathcal{X}$. Relying on such a betweenness $%
B_{\mu_{\mathcal{X}}}$, a `natural' notion of \textit{single-peakedness }for
preference preorders on $\mathcal{X}=(X,\leqslant)$ can be defined as
follows.

$\bigskip$

\begin{definition}
(\textbf{Single-peaked preference preorders on a median join-semilattice}).
Let $\mathcal{X}=(X,\leqslant)$ be a finite median join-semilattice and$\
\succcurlyeq$ a preorder i.e. a reflexive and transitive binary relation on $%
X$ (we shall denote by $\succ$ and $\sim$ its asymmetric and symmetric
components, respectively). Then, $\succcurlyeq$ is said to be \textbf{%
single-peaked \ }with respect to betweenness relation $B_{\mu _{\mathcal{X}%
}} $ (or $B_{\mu_{\mathcal{X}}}$-single-peaked) if and only if

$U$-$(i)$ there exists a \textit{unique maximum} of $\succcurlyeq $ in $X$,
its \textit{top }outcome -denoted $top(\succcurlyeq )$ - and

$U$-$(ii)$ for all $x,y,z\in X$, if $x=top(\succcurlyeq)$ and $z=\mu _{%
\mathcal{X}}(x,y,z)$ then not $y\succ z$.
\end{definition}

\bigskip

We denote by $U_{B_{\mu}}$ the set of all $B_{\mu_{\mathcal{X}}}$-\textit{%
single-peaked} preorders on $X$. An $N$-profile of $B_{\mu}$-single-peaked
preorders is a mapping from $N$ into $U_{B_{\mu}}$. We denote by $U_{B_{\mu_{%
\mathcal{X}}}}^{N}$ the set of all $N$-profiles of $B_{\mu}$-lu preorders.

Moreover, a set $D\subseteq U_{B_{\mu_{\mathcal{X}}}}^{N}$ of preorders
which are single-peaked w.r.t. $B_{\mu_{\mathcal{X}}}$ is a \textbf{rich}%
\textit{\ }\textbf{single-peaked domain }for $\mathcal{X}$\textbf{\ }if for
all $x,y\in X$ there exists $\succcurlyeq \in D$ such that $top(\succcurlyeq
)=x$ and $UC(\succeq,y)=\left \{ z\in X:z=\mu(x,y,z)\right \} $ (where $%
UC(\succeq,y):=\left \{ y\in X:x\succcurlyeq y\right \} $ is the upper
contour of $\succcurlyeq$ at $y$).

An aggregation rule $f$ \ for $(N,X)$ is \textbf{strategy-proof }\textit{on }%
$U_{B_{\mu_{\mathcal{X}}}}^{N}$ iff for all $B_{\mu_{\mathcal{X}}}$%
-single-peaked $N$-profiles $(\succcurlyeq_{i})_{i\in N}\in$ $U_{B_{\mu _{%
\mathcal{X}}}}^{N}$, and for all $i\in N$, $y_{i}\in X$, and $(x_{j})_{j\in
N}\in X^{N}$ such that $x_{j}=top(\succcurlyeq_{j})$ for each $j\in N$, 
\textit{not }$f((y_{i},(x_{j})_{j\in N\smallsetminus \left \{ i\right \}
}))\succ_{i}f((x_{j})_{j\in N})$. Finally, an aggregation rule $%
f:X^{N}\rightarrow X$ \ is $B_{\mu_{\mathcal{X}}}$-\textbf{monotonic}\textit{%
\ }iff for all $i\in N$, $y_{i}\in X$, and $(x_{j})_{j\in N}\in X^{N}$,

$f((x_{j})_{j\in N})=\mu_{\mathcal{X}}(x_{i},f((x_{j})_{j\in
N}),f(y_{i},(x_{j})_{j\in N\smallsetminus \left \{ i\right \} }))$.\footnote{%
$B_{\mu _{\mathcal{X}}}$-monotonicity of $f$ amounts to requiring all of its
projections $f_{i}$ to be \textit{gate maps }to the image of $f$ (see van de
Vel (1993), p.98 for a definition of gate maps). The introduction of $%
B_{\mu_{\mathcal{X}}}$-monotonic functions in a strategic social choice
setting is essentially due to Danilov (1994).}

In particular, let $\mathcal{X}=(X,\leqslant)$ be a finite \textit{%
join-semilattice} and $M_{\mathcal{X}}$ the set of its meet-irreducible
elements, and for any $x^{N}\in X^{N}$, and any $m\in M_{\mathcal{X}}$,
posit $N_{m}(x^{N}):=\left \{ i\in N:x_{i}\leqslant m\right \} $. Then, the
following properties of an aggregation rule can also be introduced:

$M_{\mathcal{X}}$\textbf{-Independence: }an aggregation rule $%
f:X^{N}\rightarrow X$ is $M_{\mathcal{X}}$\textbf{-independent }if and only
if for all $x_{N},y_{N}\in X^{N}$ and all $m\in M_{\mathcal{X}}$: if $%
N_{m}(x_{N})=N_{m}(y_{N})$ then $f(x_{N})\leqslant m$ if and only if $%
f(y_{N})\leqslant m$.

\textbf{Isotony: }an aggregation rule $f:X^{N}\rightarrow X$ is \textbf{%
Isotonic} if $\ f(x_{N})\leqslant f(x_{N}^{\prime})$ for all $%
x_{N},x_{N}^{\prime}\in X^{N}$ such that $x_{N}\mathbf{\leqslant}%
x_{N}^{\prime}$ (i.e. $x_{i}\leqslant x_{i}^{\prime}$ for each $i\in N$).

It can be easily shown (see Monjardet (1990)) that the \textit{conjunction}
of $M_{\mathcal{X}}$\textbf{-Independence }and \textbf{Isotony }is
equivalent to the following condition:

\textbf{\ Monotonic }$M_{\mathcal{X}}$\textbf{-Independence: }An aggregation
rule $f:X^{N}\rightarrow X$ is \textbf{monotonically }$M_{\mathcal{X}}$-%
\textbf{independent }if and only if for all $x_{N},y_{N}\in X^{N}$ and all $%
m\in M_{\mathcal{X}}$: if $N_{m}(x_{N})\subseteq N_{m}(y_{N})$ then $%
f(x_{N})\leqslant m$ implies $f(y_{N})\leqslant m$.\footnote{%
The notions of $J_{\mathcal{X}}$-Independence and Monotonic $J_{\mathcal{X}}$%
-Independence are defined similarly by dualization for a finite median
meet-semilattice $\mathcal{X}=(X,\leqslant)$ as follows: for all $%
x_{N},y_{N}\in X^{N}$ and all $j\in J_{\mathcal{X}}$, if
\par
$N_{j}(x_{N}):=\left \{ i\in N:j\leqslant x_{i}\right \} \subseteq$%
\par
$\subseteq N_{j}(y_{N}):=\left \{ i\in N:j\leqslant y_{i}\right \} $%
\par
then $j\leqslant f(x_{N})$ implies $j\leqslant f(y_{N})$.}

We are now ready to establish the following claim.

\bigskip

\begin{claim}
\textbf{\ }$(\mathcal{P}(A),\subseteq)$, $(\mathcal{R}_{A},\subseteq)$, $(%
\mathcal{R}_{B},\subseteq)$ for any $B\subseteq A$, and
\end{claim}

$\dbigcup \limits_{B\in \mathcal{P}(A)}(\mathcal{R}_{B},\subseteq)$ are 
\textit{median join-semilattices}.

\begin{proof}
Let us define the join of two total preorders $R,R^{\prime}\in \mathcal{R}%
_{A} $ as the \textit{transitive closure} $\overline{\cup}$ of their
set-theoretic union. Then, by construction, $\mathcal{X}:=(\mathcal{R}_{A},%
\overline{\cup})$ is a \textit{join-semilattice}, and satisfies both \textit{%
upper distributivity} (by Claim (P.1) of Janowitz (1984)), and \textit{%
co-coronation} (by Claims (P.3) and (P.5) of Janowitz (1984)). It follows
that $(\mathcal{R}_{A},\overline{\cup})$ thus defined is indeed a \textit{%
median join-semilattice }(whose median ternary operation is denoted here $%
\mu^{\prime}$), and its meet-irreducibles are the \textit{total preorders }$%
R_{A_{1}A_{2}}\in \mathcal{R}_{A}$ \textit{having just two (non-empty)
indifference classes} $A_{1},A_{2}$ such that (i) $(A_{1},A_{2})$ is a
two-block ordered partition of $A$, written $(A_{1},A_{2})\in \Pi _{A}^{(2)}$%
, namely $A_{1}\cup A_{2}=A$, $A_{1}\cap A_{2}=\emptyset$ and (ii) [$%
xR_{A_{1}A_{2}}y$ and \textit{not }$yR_{A_{1}A_{2}}x$] if and only if $x\in
A_{1}$ and $y\in A_{2}$. It can be easily checked that such total preorders $%
R_{A_{1}A_{2}}$ with $(A_{1},A_{2})\in \Pi_{A}^{(2)}$ are also the \textit{%
co-atoms of }$(\mathcal{R}_{A},\overline{\cup})$. Of course, the very same
argument applies to $(\mathcal{R}_{B},\overline{\cup})$, for every $%
B\subseteq A$. Moreover, the partially ordered set $\mathcal{X}^{\prime }:=(%
\mathcal{P}(A),\subseteq)$ of agendas is of course a bounded distributive
lattice with respect to set-theoretic union $\cup$ and intersection $\cap$.
Hence $(\mathcal{P}(A),\cup)$ is in particular a median join-semilattice. As
a consequence, the product join-semilattice $\mathcal{X}\times \mathcal{X}%
^{\prime}:=(\mathcal{R}_{A}\times \mathcal{P}(A),\overline{\cup}\times \cup)$
is also a \textit{median join-semilattice: }indeed, the ternary
product-operation $\mu_{\mathcal{X}}\times \mu_{\mathcal{X}^{\prime}}:(%
\mathcal{R}_{A}^{T}\times \mathcal{P}(A))^{3}\longrightarrow \mathcal{R}%
_{A}\times \mathcal{P}(A)$ inherits the characteristic median properties $%
\mu(i),\mu(ii)$ (as previously defined above) from its components. Finally, $%
\dbigcup \limits_{B\in \mathcal{P}(A)}(\mathcal{R}_{B},\subseteq)$ is a
median join-semilattice with join $\overline{\cup}$, because it is a
(disjoint) \textit{sum }(or co-product) of the family $\left \{ (\mathcal{R}%
_{B},\subseteq)\right \} _{B\subseteq A} $ of median join-semilattices, and
its median operation $\mu^{\sqcup}$ is defined as follows: for any $B,C,D\in%
\mathcal{P}(A)$, and $R^{B}\in \mathcal{R}_{B}$, $R^{C}\in \mathcal{R}_{C}$, 
$R^{D}\in \mathcal{R}_{D}$,

$\mu^{\sqcup}(R^{B},R^{C},R^{D})=(R^{B}\overline{\cup}R^{C})\cap (R^{C}%
\overline{\cup}R^{D})\cap(R^{B}\overline{\cup}R^{D})$.
\end{proof}

\bigskip

Therefore, in particular, $\mathcal{X}^{\ast}:=(\dbigcup \limits_{B\in%
\mathcal{P}(A)}\mathcal{R}_{B},\overline{\cup})$ is also endowed with a
`natural' metric $d$ (namely $d=d_{r}=\delta_{C(\mathcal{X}^{\ast})}$). But
then, any preference relation $R\in$ $\dbigcup \limits_{B\in \mathcal{P}(A)}%
\mathcal{R}_{B}$ induces in a `natural' way a reflexive preference relation $%
\mathbf{R}_{R}$ on $\dbigcup \limits_{B\in \mathcal{P}(A)}\mathcal{R}_{B}$
which has $R$ itself as its unique maximum and is \textit{single-peaked}
with respect to $d$, being defined as follows: for any $R^{\prime},R^{\prime%
\prime}$ $\in(\dbigcup \limits_{B\in \mathcal{P}(A)}\mathcal{R}%
_{B})\smallsetminus \left \{ R\right \} $, $R^{\prime}\mathbf{R}%
_{R}R^{\prime \prime}$ holds if and only if $R^{\prime}$ lies on a geodesic
from $R$ to $R^{\prime \prime}$ on the Hasse diagram $C(\mathcal{X}^{\ast})$%
. Moreover, it can be shown that any such $\mathbf{R}_{R}$ is also \textit{%
transitive}\footnote{%
See e.g. Sholander (1954), Section 3, property 3.6 for a proof.}.

Thus, it turns out that $\dbigcup \limits_{B\in \mathcal{P}(A)}\mathcal{R}%
_{B} $ can also be `naturally' endowed with a set $\mathcal{D}%
^{sp(d)}:=\left \{ \mathbf{R}_{R}:R\in \dbigcup \limits_{B\in \mathcal{P}(A)}%
\mathcal{R}_{B}\right \} $ of \textit{preorders }on $\dbigcup \limits_{B\in%
\mathcal{P}(A)}\mathcal{R}_{B}$ which are \textit{single-peaked} with
respect to the `intrinsic' metric $d$ of $\dbigcup \limits_{B\in \mathcal{P}%
(A)}\mathcal{R}_{B}$ itself.\footnote{%
Notice that $\mathcal{D}^{sp(d)}$ includes the set $\mathcal{R}^{sp(d)}$ of
all \textit{total preorders }(hence in particular all the \textit{linear
orders}) on $\dbigcup \limits_{B\in \mathcal{P}(A)}\mathcal{R}_{B}$ which
are single-peaked with respect to $d$. Moreover, $\mathcal{R}^{sp(d)}$
includes\ in turn the subclass $\mathcal{R}^{msp(d)}$ of all \textit{metric }%
single-peaked total preorders on $\dbigcup \limits_{B\in \mathcal{P}(A)}%
\mathcal{R}_{B}$ (namely those total preorders which are \textit{entirely }%
determined by the $d$-distance from the peak).}

\bigskip 

\textbf{2. Main results.}

We are now eventually ready to provide two distinct definitions of
agenda-manipulation proofness to be matched, respectively, to PAFE and SAFE
social welfare functions as defined above.

$\mathcal{\ }$

\textbf{Definition }\textit{Agenda Manipulation-Proofness of a PAFE social
welfare function (AMP}$_{P}$\textit{). }A PAFE social welfare function $%
\mathbf{f:(}\mathcal{P}(A)^{N}\times \mathcal{R}_{A})^{N}\longrightarrow 
\mathcal{P}(A)\times \mathcal{R}_{A}$ with projections $\mathbf{f}_{1},%
\mathbf{f}_{2\text{ }}$is AMP$_{P}$ if for all $i\in N$, $R_{N}\in \mathcal{R%
}_{A}^{N}$, and $B_{N},B_{N}^{\prime}\in \mathcal{P}(A)^{N}$ such that $C=%
\mathbf{f}_{1}(B_{N},R_{N})\subseteq \mathbf{f}_{1}(B_{N}^{\prime },R_{N})=D$%
,

$\mathbf{f}_{2}(B_{N},R_{N})_{|C}\mathbf{R}_{R_{i}}\mathbf{f}%
_{2}(B_{N}^{\prime},R_{N})_{|C}$ iff $\mathbf{f}_{2}(B_{N}^{%
\prime},R_{N})_{|C}\mathbf{R}_{R_{i}}\mathbf{f}_{2}(B_{N},R_{N})_{|C}$.

Accordingly, a \textit{social welfare function} $f:\mathcal{R}%
_{A}^{N}\rightarrow \mathcal{R}_{A}$ is said to be AMP$_{P}$ if the PAFE $%
\mathbf{f\simeq}f^{0}\times f$ is AMP$_{P}$ for every \textit{sovereign}
agenda formation rule $f^{0}:\mathcal{P}(A)^{N}\longrightarrow \mathcal{P}%
(A) $.

\bigskip

In other words, AMP$_{P}$ requires that at every preference profile $%
R_{N}=(R_{i})_{i\in N}$ on the entire set $A$ of admissible alternatives
each agent $i$ be \textit{indifferent }(according to her preference\ $%
\mathbf{R}_{R_{i}}$on preference preorders on $A$ as induced by her actual
preference $R_{i}$ on $A$) between the \textit{restriction} of the \textit{%
social} preference $f(R_{N})$ to an arbitrary agenda $C$, \textit{no matter
if that agenda is the actually selected agenda }$D$ \textit{or just a
subagenda of }$D$.

\bigskip

\textbf{Definition }\textit{Agenda Manipulation-Proofness of a SAFE social
welfare function (AMP}$_{S}$\textit{) }

A SAFE social welfare function\textit{\ }$\widehat{\mathbf{f}}=(f^{0},%
\mathcal{F}(f^{0}))$ (with

$\mathcal{F}(f^{0})=\left \{ f_{B}^{0}:\mathcal{R}_{B}^{N}\rightarrow 
\mathcal{R}_{B}\right \} _{B\in f^{0}[\mathcal{P}(A)]}$as defined above) is
AMP$_{S}$ if for all $i\in N$, $R_{N}\in \mathcal{R}_{A}^{N}$, and $%
B_{N},B_{N}^{\prime}\in \mathcal{P}(A)^{N}$, $C,D\in \mathcal{P}(A)$ such
that $C=f^{0}(B_{N})\subseteq f^{0}(B_{N}^{\prime})=D$, $\ $

$f_{C}^{0}((R_{N})_{|C})\mathbf{R}_{R_{i}}(f_{D}^{0}((R_{N})_{|D})_{|C}$ iff 
$\ f_{D}^{0}((R_{N})_{|D})_{|C}\mathbf{R}_{R_{i}}\ f_{C}^{0}((R_{N})_{|C})$.

Accordingly, a \textit{social welfare function} $f:\mathcal{R}%
_{A}^{N}\rightarrow \mathcal{R}_{A}$ is AMP$_{S}$ iff for every \textit{%
sovereign} agenda formation rule $f^{0}:\mathcal{P}(A)^{N}\longrightarrow%
\mathcal{P}(A)$ the SAFE social welfare function $\widehat{\mathbf{f}}%
=(f^{0},\mathcal{F})$ with \textit{uniform }family $\mathcal{F}=\left \{
(f_{A})_{|B}:\mathcal{R}_{B}^{N}\rightarrow \mathcal{R}_{B}\right \} _{B\in%
\mathcal{P}(A)}$ induced by $f$\ is AMP$_{S}$.

\bigskip

Thus, AMP$_{S}$ requires that at every preference profile $%
R_{N}=(R_{i})_{i\in N}$ on the entire set $A$ of admissible alternatives
each agent $i$ be \textit{indifferent }(according to her preference\ $%
\mathbf{R}_{R_{i}}$on preference preorders on $A$ as induced by her actual
preference $R_{i}$ on $A$) between the \textit{social} preference $%
f_{C}^{0}((R_{N})_{|C})$ at the \textit{restriction} of $R_{N}$ to any
selected agenda $C$, and the \textit{restriction }to $C$ of the \textit{%
social} preference $f_{D}^{0}((R_{N})_{|D})$ at the \textit{restriction} of $%
R_{N}$ to any other selected agenda $D\supseteq C$.\textit{\ }

As mentioned above, in his classic work (Arrow (1963)) Arrow refers to the
need to prevent agenda manipulation as the main argument to support the
requirement of Independence of Irrelevant Alternatives for social welfare
functions, that is defined as follows.

\bigskip

\textbf{Definition }\textit{Independence of Irrelevant Alternatives (IIA)}.

A social welfare function $f_{A}:\mathcal{R}_{A}^{N}\rightarrow \mathcal{R}%
_{A}$ satisfies IIA iff for all $R_{N},R_{N}^{\prime}\in \mathcal{R}_{A}^{N} 
$, and $B\in \mathcal{P}(A)$, $(f(R_{N}))_{|B}=(f(R_{N}^{\prime}))_{|B}$
whenever $(R_{N})_{|B}=(R_{N}^{\prime})_{|B}$.

\bigskip

Therefore, we have just introduced \textit{three }distinct conditions that
are meant to address the \textit{same} problem, namely preventing agenda
manipulation. A first fact about such conditions is worth mentioning at the
outset: when regarded as conditions on social welfare functions both AMP$%
_{P} $ and AMP$_{S}$ only make reference to an arbitrary \textit{single} 
\textit{preference profile} on $A$, while IIA concerns an arbitrary \textit{%
pair of} \textit{preference profiles }on $A$. That contrast is quite
remarkable, because reference to a \textit{single} preference profile is a
feature that seems to make full sense, in view of Arrow's overt intention to
put aside all the issues related to possible strategic misrevelation of
preferences. Notice, however, that in Arrow's work the notion of agenda
manipulation-proofness is only introduced in a quite informal way.
Accordingly, our next task is to explore the precise relationship of IIA to
each one of the agenda manipulation-proofness properties introduced above.

Let us start from AMP$_{P}$. Indeed, our first finding is that IIA is 
\textit{not at all} related to AMP$_{P}$.

\bigskip

\begin{proposition}
Let $f:\mathcal{R}_{A}^{N}\rightarrow \mathcal{R}_{A}$ be a social welfare
function and $\mathbf{f}:(\mathcal{R}_{A}\times \mathcal{P}%
(A))^{N}\rightarrow \mathcal{R}_{A}\times \mathcal{P}(A)$ a decomposable PAFE
social welfare function for $(N,A)$ such that $\mathbf{f\simeq}f^{0}\times f$
where\ $f^{0}$ is an arbitrary sovereign agenda-formation rule. Then, $%
\mathbf{f}$ is AMP$_{P}$ (and consequently $f$ is also AMP$_{P}$, by
definition).
\end{proposition}

\begin{proof}
Straightforward: let $R_{N}\in \mathcal{R}_{A}^{N}$, and $%
B_{N},B_{N}^{\prime }\in \mathcal{P}(A)^{N}$ such that $C=\mathbf{f}%
_{1}(B_{N},R_{N})\subseteq \mathbf{f}_{1}(B_{N}^{\prime},R_{N})=D$. By
decomposability of $\mathbf{f}$, $\mathbf{f}_{2}(B_{N},R_{N})=\mathbf{f}%
_{2}(B_{N}^{\prime},R_{N})=f(R_{N})$. Hence, for every $i\in N$, \textbf{\ }%
both $\mathbf{f}_{2}(B_{N},R_{N})_{|C}\mathbf{R}_{R_{i}}\mathbf{f}%
_{2}(B_{N}^{\prime},R_{N})_{|C}$ and

$\mathbf{f}_{2}(B_{N}^{\prime},R_{N})_{|C}\mathbf{R}_{R_{i}}\mathbf{f}%
_{2}(B_{N},R_{N})_{|C}$ hold by reflexivity of $\mathbf{R}_{R_{i}}$, and the
thesis follows.
\end{proof}

\bigskip

Observe that, when formally considered as a condition for a PAFE social
welfare function $\mathbf{f}$, AMP$_{P}$ is in fact an \textit{interprofile }%
condition\footnote{%
See Fishburn (1973) for a careful classification of structural, interprofile
and intraprofile conditions for social welfare functions and related
constructs.} because it involves \textit{two} profiles $%
(B_{N},R_{N}),(B_{N}^{\prime},R_{N})$ in $(\mathcal{R}_{A}\times \mathcal{P}%
(A))^{N}$. However, the projection of AMP$_{P}$ to the social welfare
component $f$ of $\mathbf{f}$ collapses in fact to an \textit{intraprofile}
condition since it involves a single profile $R_{N}$ $\in \mathcal{R}%
_{A}^{N} $. Now, IIA is of course an \textit{interprofile} condition for
social welfare functions involving arbitrary \textit{pairs} of profiles in $%
\mathcal{R}_{A}^{N}$. Therefore, ostensibly, AMP$_{P}$ and IIA are \textit{%
mutually unrelated} as conditions for social welfare functions.

Let us now turn to the relationship between AMP$_{s}$ of a SAFE social
welfare function and IIA. In order to accomplish that task, we shall take
advantage of the notion of \  \textit{projectivity }of a preference profile
for a social welfare function. Indeed, let $B\subseteq A$, $f:\mathcal{R}%
_{B}^{N}\rightarrow \mathcal{R}_{B}$ and $R_{N}\in \mathcal{R}_{B}^{N}$:
then, $R_{N}$ is a\textit{\ projective profile }for $f$ if there exists $%
i\in N$ such that $f(R_{N})=R_{i}$. Next, we introduce a considerably
weakened version of IIA, namely its restriction to projective profiles in
the following sense:

\bigskip

\textit{Independence of Irrelevant Alternatives at Projective Profiles (IIAP)%
}.

A social welfare function $f:\mathcal{R}_{A}^{N}\rightarrow \mathcal{R}_{A}$
satisfies IIAP iff for all $B\in \mathcal{P}(A)$ and $R_{N},R_{N}^{\prime}\in%
\mathcal{R}_{A}^{N}$, if $R_{N}$ and $R_{N}^{\prime}$ are projective
profiles for $f$ and $(R_{N})_{|B}=(R_{N}^{\prime})_{|B}$ then $%
(f(R_{N}))_{|B}=(f(R_{N}^{\prime}))_{|B}$.

\bigskip

Observe that IIAP is indeed \textit{strictly weaker} than IIA. To check the
validity of that statement just consider the social welfare function $%
f_{i^{\ast}}^{BC}$for $(N,A)$ defined as follows: for $R_{N}\in \mathcal{R}%
_{A}^{N}$, $f_{i^{\ast}}^{BC}(R_{N}):=f^{BC}(R_{N})$ (where $f^{BC}$ denotes
the Borda-Count scoring aggregation rule) if $R_{N}$ is \textit{not} a
projective profile for $f^{BC}$ and $f_{i^{\ast}}^{BC}(R_{N}):=R_{i^{\ast}}$
otherwise. By construction, $R_{N}\in \mathcal{R}_{A}^{N}$ is projective for 
$f_{i^{\ast}}^{BC}$ if and only if $f_{i^{\ast}}^{BC}(R_{N})=R_{i^{\ast}}$%
hence $f_{i^{\ast}}^{BC}$ does satisfy IIAP. However, $f_{i^{\ast}}^{BC}$
clearly violates IIA:\ it is easily checked that there exist profiles $%
R_{N},R_{N}^{\prime}$ that are not projective for $f^{BC}$ and such that $%
(f_{i^{\ast}}^{BC}(R_{N}))_{|B}\neq(f_{i^{\ast}}^{BC}(R_{N}^{\prime}))_{|B}$
for some $B\subseteq A$.

We are now ready to show that AMP$_{S}$ is in fact \textit{tightly }%
connected to IIA, as established by the following proposition.

\bigskip \ 

\begin{proposition}
Let $f:\mathcal{R}_{A}^{N}\rightarrow \mathcal{R}_{A}$ be a social welfare
function for $(N,A)$. Then, (i) if \ $f$ satisfies IIA then $f$ is also AMP$%
_{S}$; ii) if $f$ is AMP$_{S}$ then $f$ satisfies IIAP. $\  \  \ $
\end{proposition}

\begin{proof}
(i) Let $f:\mathcal{R}_{A}^{N}\rightarrow \mathcal{R}_{A}$ be a social
welfare function for $(N,A)$ that satisfies $IIA$. Then, take any sovereign
agenda-formation rule $f^{0}:\mathcal{P}(A)^{N}\longrightarrow \mathcal{P}%
(A) $ and consider the \textit{uniform} SAFE social welfare function $%
\widehat {\mathbf{f}}=(f^{0},\mathcal{F}(f^{0}))$ induced by $f$, namely
with $\mathcal{F}(f^{0}):=\left \{ f_{B}:\mathcal{R}_{B}^{N}\rightarrow 
\mathcal{R}_{B}\right \} _{B\subseteq A}$ defined as follows: for every $%
B\subseteq A$, and $Q_{N}\in \mathcal{R}_{B}^{N}$, $%
f_{B}(Q_{N}):=(f(R_{N}))_{|B}$ where $R_{N}\in \mathcal{R}_{A}$ is such that 
$(R_{N})_{|B}=Q_{N}$. Clearly, such an $f_{B}$ is well-defined precisely
because $f$ satisfies IIA. But then, take any $R_{N}\in \mathcal{R}_{A}^{N}$%
, and $B_{N},B_{N}^{\prime}\in \mathcal{P}(A)^{N}$, $C,D\in \mathcal{P}(A)$
such that $C=f^{0}(B_{N})\subseteq f^{0}(B_{N}^{\prime})=D$, and consider $%
f_{C}((R_{N})_{|C})$\textbf{\ }and $(f_{D}((R_{N})_{|D}))_{|C}$. By
definition $%
f_{C}((R_{N})_{|C})=(f(R_{N}))_{|C}=(f(R_{N})_{|D})_{|C}=(f_{D}((R_{N})_{|D}))_{|C} 
$ whence $\ $

$f_{C}((R_{N})_{|C})\mathbf{R}_{R_{i}}(f_{D}((R_{N})_{|D}))_{|C}$ iff $\
(f_{D}((R_{N})_{|D}))_{|C}\mathbf{R}_{R_{i}}\ f_{C}((R_{N})_{|C})$ for all $%
i\in N$ i.e. $\widehat{\mathbf{f}}$ satisfies AMP$_{S}$ hence by definition $%
f$ is also AMP$_{S}$.

(ii) Suppose that $f$ \ is AMP$_{S}$ yet it violates IIAP. Thus, there exist 
$R_{N},R_{N}^{\prime}\in \mathcal{R}_{A}^{N}$ and $B\subseteq A$ such that $%
R_{N}$,$R_{N}^{\prime}$ are projective profiles for $f$ and $%
(R_{N})_{|B}=(R_{N}^{\prime})_{|B}$ yet $(f(R_{N}))_{|B}\neq(f(R_{N}^{%
\prime}))_{|B}$, hence $f(R_{N})\neq f(R_{N}^{\prime})$. Moreover, by
projectivity of $R_{N}$ and $R_{N}^{\prime}$, there exist $i,j\in N$ such
that $R_{i}=f(R_{N})$, $R_{j}^{\prime}=f(R_{N}^{\prime})$ whence, in
particular, $R_{i|B}=(f(R_{N}))_{|B}\neq(f(R_{N}^{\prime}))_{|B}=R_{j|B}^{%
\prime}$. However, $(R_{N})_{|B}=(R_{N}^{\prime})_{|B}$ implies that $%
R_{i|B}=R_{i|B}^{\prime}$ and $R_{j|B}=R_{j|B}^{\prime}$. Now,\ let $f^{0}$
be a sovereign agenda formation rule and $\widehat{\mathbf{f}}=(f^{0},%
\mathcal{F})$ a \textit{uniform} SAFE social welfare function with $\mathcal{%
F}=\left \{ f_{B}:\mathcal{R}_{B}^{N}\rightarrow \mathcal{R}_{B}|\text{ }%
f_{B}=(f_{A})_{|B}\right \} _{B\in \mathcal{P}(A)}$ and $f=f_{A}\in \mathcal{%
F}$ . Clearly $f_{B}((R_{N})_{|B})=f_{B}((R_{N}^{\prime})_{|B})$.

But then, either $(f(R_{N}))_{|B}\neq$ $f_{B}((R_{N})_{|B})$ or $%
(f(R_{N}^{\prime}))_{|B}\neq$ $f_{B}((R_{N}^{\prime})_{|B})$.

Suppose w.l.o.g. that $(f(R_{N}))_{|B}\neq$ $f_{B}((R_{N})_{|B})$. Then, by
definition of $\mathbf{R}_{R_{i}}$, $R_{i|B}=(f(R_{N}))_{|B}$ implies

$(f(R_{N}))_{|B}\mathbf{R}_{R_{i}}(f_{B}((R_{N})_{|B\text{ }}))$ and \textit{%
not \ }$(f_{B}((R_{N})_{|B\text{ }}))\mathbf{R}_{R_{i}}(f(R_{N}))_{|B\text{ }%
}$

whence AMP$_{S}$ fails, a contradiction.
\end{proof}

\bigskip 

\begin{remark}
One might also wonder whether IIA itself is also a necessary condition for
any social welfare function $f$ to be AMP$_{S}$. But that is clearly not the
case. To see that, just consider for any $R^{\ast}\in \mathcal{R}_{A}$, $%
\varnothing \neq B^{\ast}\subset A$ and $i\in N$, the social welfare
function $f^{iR^{\ast}B^{\ast}}$defined as follows: for every $R_{N}\in 
\mathcal{R}_{A}^{N}$,

$f^{iR^{\ast}B^{\ast}}(R_{N}):=\left \{ 
\begin{array}{c}
R^{\ast}\text{ if }R_{j|A\diagdown B^{\ast}}=R_{|A\setminus B^{\ast}}^{\ast }%
\text{ for some }j\in N\setminus \left \{ i\right \} \\ 
\text{and }R_{i}\text{ otherwise}%
\end{array}
\right \} $.

It is easy to check that $f^{iR^{\ast}B^{\ast}}$is AMP$_{S}$ because for any 
$C\subseteq D\subseteq A$, and $R_{N}\in \mathcal{R}_{A}^{N}$, $(f^{iR^{\ast
}B^{\ast}}(R_{N}))_{|C}\mathbf{=}((f^{iR^{\ast}B^{\ast}}(R_{N}))_{|D})_{|C}$.

However, $f^{iR^{\ast }B^{\ast }}$violates IIAP (hence, in particular, IIA
as well).\ Indeed, consider profiles $R_{N},R_{N}^{\prime }\in \mathcal{R}%
_{A}^{N}$ such that $R_{j|B^{\ast }}=R_{j|B^{\ast }}^{\prime }$ for all $%
j\in N$, $R_{i|B^{\ast }}\neq R_{|B^{\ast }}^{\ast }$, $R_{j|A\setminus
B^{\ast }}\neq R_{|A\setminus B^{\ast }}^{\ast }$ for each $j\in N\setminus
\left \{ i\right \} $, and $R_{k}^{\prime }=R^{\ast }$for some $k\in
N\setminus \left \{ i\right \} $. Now, both $R_{N}$ and $R_{N}^{\prime }$
are projective profiles for $f^{iR^{\ast }B^{\ast }}$since $f^{iR^{\ast
}B^{\ast }}(R_{N})=R_{i}\neq R^{\ast }$ while $f^{iR^{\ast }B^{\ast
}}(R_{N}^{\prime })=R^{\ast }=R_{k}^{\prime }$. In particular, $(f^{iR^{\ast
}B^{\ast }}(R_{N}))_{|B^{\ast }}$ $=R_{i|B^{\ast }}\neq R_{|B^{\ast }}^{\ast
}=(f^{iR^{\ast }B^{\ast }}(R_{N}^{\prime }))_{|B^{\ast }}$ though $%
(R_{N})_{|B^{\ast }}=(R_{N}^{\prime })_{|B^{\ast }}$, hence IIAP is violated.
\end{remark}

Finally, we can proceed to the next main task of the present analysis, which
is to explore the class of social welfare functions which are agenda
manipulation-proof \textit{and }do satisfy at least some \textit{minimal}
combination of \textit{outcome-unbiasedness} and \textit{distributed
responsiveness }to agents' preferences, as specified below.

A basic unbiasedness requirement is embodied in the standard sovereignty
property, as defined below.

\bigskip

\textbf{Sovereignty (S) }A social welfare function $f:\mathcal{R}%
_{A}^{N}\rightarrow \mathcal{R}_{A}$ for $(N,A)$ is \textbf{sovereign }if
for each $R\in \mathcal{R}_{A}$ there exists $R_{N}\in \mathcal{R}_{A}^{N}$
such that $f(R_{N})=R$.

\bigskip

A further, and weaker, unbiasedness condition is implicit in the following
property.

\bigskip

\textbf{Weak Sovereignty (WS) }A social welfare function $f:\mathcal{R}%
_{A}^{N}\rightarrow \mathcal{R}_{A}$ for $(N,A)$ is \textbf{weakly sovereign 
}if for any $x,y\in A$ there exists $R_{N}\in \mathcal{R}_{A}^{N}$ such that 
$xf(R_{N})y$.

\bigskip

Clearly, WS ensures a \textit{minimal degree of outcome-unbiasedness and
responsiveness }of the relevant social welfare function, but it is
consistent both with fairly distributed responsiveness-patterns involving a
large number of agents, and with extremely concentrated
responsiveness-patterns involving very few agents, or even just a single
agent.

In order to make precise such distributed responsiveness requirement, we
introduce the \textit{responsiveness correspondence }of a social welfare
function as defined below.

\bigskip

\textbf{Responsiveness Correspondence of a social welfare function} Let $f:%
\mathcal{R}_{A}^{N}\rightarrow \mathcal{R}_{A}$ be a social welfare function
for $(N,A)$. Then, its \textit{responsiveness correspondence }$F_{f}:A\times
A\rightrightarrows \mathcal{P}(N)$ is defined as follows: for every $x,y\in
A $,

$F_{f}(x,y):=\left \{ 
\begin{array}{c}
S\subseteq N:\text{there exists }R_{S}^{xy}\in \mathcal{R}_{A}^{S}\text{
such that for all }R_{N}\in \mathcal{R}_{A}^{N}\text{, } \\ 
\text{if\ }[xR_{i}y\text{ iff }xR_{i}^{xy}y\text{ for every }i\in S]\text{
then }xf(R_{N})y%
\end{array}
\right \} $.

\bigskip

\textbf{Minimally Distributed Responsiveness (MDR) }A social welfare
function $f:\mathcal{R}_{A}^{N}\rightarrow \mathcal{R}_{A}$ for $(N,A)$
satisfies \textbf{minimally distributed responsiveness }if whenever $\left
\{ i\right \} \in F_{f}(x,y)$ for some $i\in N$ and some pair of \textit{%
distinct} $x,y\in A$ it must be the case that there exist $S\subseteq
N\setminus \left \{ i\right \} $ and $v,z\in A$, $v\neq z$ such that $S\in
F_{f}(v,z)$.

\bigskip

In plain words, if the nonstrict preference of a single agent $i$ between
two distinct alternatives $x,y$ has to be accepted as part of the social
preference, then the nonstrict preference between two distinct alternatives $%
v$ and $z$ of \textit{some other coalition not including }$i$ is also
entitled to acceptance as part of the social preference. Thus, arguably, the 
\textit{combination} of WS and MDR amounts in fact to an appropriate \textit{%
minimal }requirement of \textit{unbiased distributed responsiveness}.

Let us now recall a few (mostly classic) requirements for social welfare
functions.

A social welfare function $f:\mathcal{R}_{A}^{N}\rightarrow \mathcal{R}_{A}$
satisfies

\textbf{Anonymity (AN) }if for every $R_{N}\in \mathcal{R}_{A}^{N}$ and
every permutation $\sigma$ of $N$, $f(x^{N})=f(R_{\sigma(N)})$ (where $%
R_{\sigma (N)}=(R_{\sigma(1)},...,R_{\sigma(n)})$);

\textbf{Idempotence (ID) }\textit{\ }if \ $f(R_{N})=R$ whenever $R_{N}$ is
such that $R_{i}=R$ for each $i\in N$;

\textbf{Neutrality (NT) }if [$xf(R_{N})y$ iff $yf(R_{N}^{\prime})x$] for any 
$x,y\in A$ and $R_{N},R_{N}^{\prime}\in \mathcal{R}_{A}^{N}$ such that $%
xR_{i}y$ iff $yR_{i}^{\prime}x$ for each $i\in N$;

\textbf{Weak Neutrality} (\textbf{WNT}) if [$f(R_{N})\subseteq R$ iff $%
f(R_{N})\subseteq R^{\prime}$] for any two-indifference-class $R,R^{\prime
}\in \mathcal{R}_{A}$ and $R_{N}\in \mathcal{R}_{A}^{N}$ such that $%
R_{i}\subseteq R$ iff $R_{i}\subseteq R^{\prime}$ for every $i\in N$;

\textbf{Weak Pareto Principle (WP) }if for every $x,y\in A$ and $R_{N}\in%
\mathcal{R}_{A}^{N}$, if $xP(R_{i})y$ for every $i\in N$ then $xP(f(R_{N}))y$%
;

\textbf{Basic Pareto Principle (BP) }if for every $x,y\in A$ and $R_{N}\in%
\mathcal{R}_{A}^{N}$, if $xR_{i}y$ for every $i\in N$ then $xf(R_{N})y$;

\textbf{Local Separation (LS) }if for every $x,y\in A$ there exist $%
R_{N},R_{N}^{\prime}\in \mathcal{R}_{A}^{N}$ such that $f(R_{N})_{|\left \{
x,y\right \} }\neq f(R_{N}^{\prime})_{_{|\left \{ x,y\right \} }}$.

It should be emphasized, for future reference, that LS implies WS, while WP
and LS are mutually independent.

Moreover, for any domain $\mathbf{D}$ of (preference) preorders on $\mathcal{%
R}_{A}$, a social welfare function $f:\mathcal{R}_{A}^{N}\rightarrow 
\mathcal{R}_{A}$ is \textbf{strategy-proof }on $\mathbf{D}$ iff for every $%
i\in N$, $\mathbf{R}_{N}\in \mathbf{D}^{N}$, $R_{N}\in$\textbf{\ }$\mathcal{R%
}_{A}^{N}$ and $R^{\prime}\in \mathcal{R}_{A}$, $\ f(R_{N})\mathbf{R}%
_{i}f((R_{i}^{\prime},R_{N\smallsetminus \left \{ i\right \} })$.

Moreover, a social welfare function $f:\mathcal{R}_{A}^{N}\rightarrow 
\mathcal{R}_{A}$ $\ $\ is said to be

\textbf{dictatorial }(respectively, \textbf{inversely dictatorial}) if there
exists $i\in N$ such that for all $R_{N}\in \mathcal{R}_{A}^{N}$ and $x,y\in
A $, $xf(R_{N})y$ only if $xR_{i}y$ (respectively, $yR_{i}x$), \textbf{%
weakly paretian }if it satisfies \textbf{WP, }and \textbf{weakly
anti-paretian }if $yP(f(R_{N}))x$ for every $x,y\in A$ and $R_{N}\in 
\mathcal{R}_{A}^{N}$ such that $xP(R_{i})y$ for every $i\in N$, \textbf{%
consensual }if it satisfies \textbf{ID}, and \textbf{properly consensual }if
it satisfies \textbf{ID} and \textbf{MDR}.

The \textbf{global stalemate }social welfare function for $(N,A)$ is the 
\textit{constant }function $f^{U_{A}}:\mathcal{R}_{A}^{N}\rightarrow 
\mathcal{R}_{A}$ such that $f^{U_{A}}(R_{N})=U_{A}$ for every $R_{N}\in$ $(%
\mathcal{R}_{A}^{T})^{N}$, where $U_{A}:=A\times A$, the \textit{universal
indifference }relation.

\bigskip

We are now ready to establish which social welfare functions among those
that ensure at least a modicum of unbiased distributed responsiveness do
also satisfy the agenda manipulation-proofness requirements AMP$_{P}$ and AMP%
$_{S}$, respectively.

Concerning AMP$_{P}$ social welfare functions, we can rely on the following
recent result (see Savaglio, Vannucci (2019, 2021)).

\bigskip

\begin{theorem}
(Savaglio, Vannucci (2021)) Let $\mathcal{X}=(X,\leqslant)$ be a finite 
\textit{median join-semilattice, }$M_{\mathcal{X}}$ the set of its
meet-irreducible elements, $B_{\mu}\mathcal{\ }$its median-induced
betweenness, and \ $f:X^{N}\rightarrow X$ \ an aggregation rule. Then, the
following statements are equivalent:

(i) \ $f$ is strategy-proof on $D^{N}$ for every rich domain $D\subseteq
U_{B_{\mu}}$ of locally unimodal preorders on w.r.t. $B_{\mu}$ on $X$;

(ii) $f$ is monotonically $M_{\mathcal{X}}$-independent;\textbf{\ }

(ii) there exists a family $\mathcal{F}_{M_{\mathcal{X}}}=\left \{
F_{m}:m\in M_{\mathcal{X}}\right \} $ of $\ $order filters of $(\mathcal{P}%
(N),\subseteq) $ such that $\ $

$f(x_{N})=f_{\mathcal{F}_{M_{\mathcal{X}}}}(x_{N}):=\tbigwedge \left \{ m\in
M_{\mathcal{X}}:N_{m}(x_{N})\in F_{m}\right \} $ for all $x_{N}\in X^{N}$ .
\end{theorem}

\bigskip 

\begin{remark}
Thus, in particular, let $\mathcal{X}=(\mathcal{R}_{A}^{T},\overline{\cup})) 
$ be the join-semilattice of total preorders on finite set $A$, $\mu$ its
median ternary operation and $B_{\mu}$ the corresponding betweenness as
previously defined, $\Pi_{A}^{(2)}$ the set of \ all total preorders on $A$
with two indifference classes and $f:\mathcal{R}_{A}^{N}\rightarrow \mathcal{%
R}_{A}$ an aggregation rule (namely, a social welfare function) for $(N,%
\mathcal{R}_{A})$. Then, $M_{\mathcal{X}}=\Pi_{A}^{(2)}$ and $f$ is
strategy-proof on $D^{N}$ for every rich domain $D\subseteq U_{B_{\mu}}$ of
locally unimodal preorders w.r.t. $B_{\mu}$ on $\mathcal{R}_{A}$ iff there
exists a family $\mathcal{F}_{M_{\mathcal{X}}}=\left \{
F_{A_{1}A_{2}}:(A_{1},A_{2})\in \Pi_{A}^{(2)}\right \} $ of order filters of 
$(\mathcal{P}(N),\subseteq)$ such that $\  \ $
\end{remark}

$\  \  \  \  \  \ f(R_{N})=f_{\mathcal{F}_{M_{\mathcal{X}}}}(R_{N}):=\tbigwedge
\left \{ R_{A_{1}A_{2}}\in M_{\mathcal{X}}:\left \{ i\in N:R_{i}\subseteq
R_{A_{1}A_{2}}\right \} \in F_{A_{1}A_{2}}\right \} $ for all $R_{N}\in%
\mathcal{R}_{A}^{N}$ .

\bigskip 

\begin{remark}
Hence, the collection of such strategy-proof social welfare functions
includes the following subclasses:

\begin{itemize}
\item \textit{Inclusive quorum systems }, namely functions $f_{\mathcal{F}%
_{M_{\mathcal{X}}}}$ such that every order filter $F_{R_{A_{1}A_{2}}}$ is 
\textit{transversal }i.e. $S\cap T\neq \varnothing$ for all $S,T\in
F_{R_{A_{1}A_{2}}}$ and $\dbigcup \limits_{R_{A_{1}A_{2}}\in M_{\mathcal{X}%
}}F_{R_{A_{1}A_{2}}}^{\min}=N$ (observe that such a\ class includes any rule
such that for every $R_{A_{1}A_{2}}\in M_{\mathcal{X}}$, $F_{R_{A_{1}A_{2}}}$
is \textit{simple-majority collegial }i.e. there exists a \textit{minimal}
simple majority coalition $S_{A_{1}A_{2}}\subseteq N$, $|S_{A_{1}A_{2}}|=%
\left \lfloor \frac{|N|+2}{2}\right \rfloor $ with $F_{R_{A_{1}A_{2}}}=\left
\{ T\subseteq N:S_{A_{1}A_{2}}\subseteq T\right \} $). Generally speaking,
inclusive quorum systems need not be anonymous or neutral.

\item \textit{Outcome-biased aggregation rules}, namely functions $f_{%
\mathcal{F}_{M_{\mathcal{X}}}}$ where $F_{R_{A_{1}A_{2}}}=\varnothing$ for
some $R_{A_{1}A_{2}}\in M_{\mathcal{X}}$ (observe that they include the
subclass of those aggregation rules such that for some \textit{total
preorder }$\overline{R}\in \mathcal{R}_{A}$,\textit{\ }including possibly a%
\textit{\ linear order, }$F_{R_{A_{1}A_{2}}}=\varnothing$ for every $%
R_{A_{1}A_{2}}\in M_{\mathcal{X}}$ such that $\overline{R}\subseteq
R_{A_{1}A_{2}}$).

\item \textit{Weakly-neutral aggregation rules}, namely functions $f_{%
\mathcal{F}_{M_{\mathcal{X}}}}$ where $F_{R_{A_{1}A_{2}}}=F_{R_{A_{1}^{%
\prime}A_{2}^{\prime}}}$ whenever $R_{A_{1}A_{2}}\wedge R_{A_{1}^{\prime
}A_{2}^{\prime}}$ exists.

\item \textit{Quota aggregation rules}, i.e. functions $f_{\mathcal{F}_{M_{%
\mathcal{X}}}}$ such that for each $R_{A_{1}A_{2}}\in M_{\mathcal{X}}$ there
exists an integer $q_{[R_{A_{1}A_{2}}]}\leq|N|$ with

$F_{m}=\left \{ T\subseteq N:q_{[R_{A_{1}A_{2}}]}\leq|T|\right \} $ (such
rules are clearly anonymous, but not necessarily weakly-neutral: they are of
course weakly-neutral as well if, furthermore, $%
F_{R_{A_{1}A_{2}}}=F_{R_{A_{1}^{\prime}A_{2}^{\prime}}}$ whenever $%
R_{A_{1}A_{2}}\wedge R_{A_{1}^{\prime}A_{2}^{\prime}}$ exists). Quota
aggregation rules are said to be \textit{positive }if $%
q_{[R_{A_{1}A_{2}}]}>0 $ for every $R_{A_{1}A_{2}}\in M_{\mathcal{X}}$. The
subclass of positive and weakly-neutral quota aggregation rules includes as
a prominent example the \textit{co-majority} social welfare function $%
f^{\partial maj}$defined as follows: for every $R_{N}\in$\textbf{\ }$%
\mathcal{R}_{A}$

\ $f^{\partial maj}(R_{N}):=\wedge_{S\in \mathcal{W}^{maj}}(\vee_{i\in
S}R_{i})$

where $\mathcal{\mathcal{W}}^{maj}:\mathcal{=}\left \{ S\subseteq N:|S|\geq%
\frac{n+1}{2}\right \} $.

\item The global stalemate social welfare function $f^{U_{A}}$ for $(N,A)$
which obtains when $F_{m}=\varnothing$ for all $m\in M_{\mathcal{X}}$\textit{%
\ .\ }
\end{itemize}
\end{remark}

\bigskip

It is worth noticing that a large subclass of such social welfare functions $%
f_{\mathcal{F}_{M_{\mathcal{X}}}}$ (including \textit{positive quota
aggregation rules} and \textit{inclusive quorum systems}) satisfy the Basic
Pareto Principle (BP), as made precise by the following claim.

\bigskip

\begin{claim}
Let $f_{\mathcal{F}_{M_{\mathcal{X}}}}$ be a social welfare function as
defined above such that $F_{R_{A_{1}A_{2}}}$is a nontrivial proper order
filter (i.e. $\varnothing \notin F_{R_{A_{1}A_{2}}}\neq \varnothing$) for
every $R_{A_{1}A_{2}}\in M_{\mathcal{X}}$. Then $f_{\mathcal{F}_{M_{\mathcal{%
X}}}}$ satisfies BP.
\end{claim}

\begin{proof}
Suppose that $x,y\in A$ and $R_{N}\in \mathcal{R}_{A}^{N}$ are such that $%
xR_{i}y$ for every $i\in N$\textit{, yet not }$xf_{\mathcal{F}_{M_{\mathcal{X%
}}}}y$. Namely, by construction,

$(x,y)\notin \tbigwedge \left \{ R_{A_{1}A_{2}}\in M_{\mathcal{X}}:\left \{
i\in N:R_{i}\subseteq R_{A_{1}A_{2}}\right \} \in F_{A_{1}A_{2}}\right \} $%
\textit{. }

Hence, there exists\textit{\ }$R_{A_{1}A_{2}}\in M_{\mathcal{X}}$ such that $%
\left \{ i\in N:R_{i}\subseteq R_{A_{1}A_{2}}\right \} \in F_{A_{1}A_{2}}$ and 
$(x,y)\notin R_{A_{1}A_{2}}$. However, by assumption , $F_{A_{1}A_{2}}$ is
nonempty and \textit{every }$T\in F_{A_{1}A_{2}}$ is itself nonempty: thus, $%
N\in F_{A_{1}A_{2}}$. But then $(x,y)\in R_{i}\subseteq R_{A_{1}A_{2}}$ for
any $i\in T$, a contradiction.
\end{proof}

\bigskip 

\begin{remark}
It should be emphasized that BP and WP are \textit{independent }alternative
ways of weakening the (strong) Pareto principle\footnote{%
A social welfare function $f:\mathcal{R}_{A}^{N}\longrightarrow \mathcal{R}%
_{A}$ satisfies the \textit{(strong) Pareto principle} iff $xP(f(R_{N}))y$
for any $x,y\in A$ and $R_{N}\in $ $\mathcal{R}_{A}^{N}$ such that $xR_{i}y$
for every $i\in N$, and $xP(R_{j})y$ for some $j\in N$.}. To see that, just
consider social welfare functions $f^{UN}$ and $f$ $^{\mathbf{L}x^{\ast }}$%
for $(N,A)$ defined informally as follows: if $R_{N}$ is such that $R_{i}=R$
for each $i\in N$ then $f^{UN}(R_{N})=R$, otherwise $f^{UN}(R_{N})$ is the
universal indifference relation on $A$ while $f^{\text{ }\mathbf{L}}x^{\ast
}(R_{N})$ -where $\mathbf{L}$ is a linear order of $\mathcal{R}_{A}$- is the 
$\mathbf{L}$-minimum linear order $L$ of $A$ having $x^{\ast }$as its top
element and such that $L\supseteq \cap _{i\in N}P(R_{i})$ if there is no $y$
such that $(y,x^{\ast })\in \cap _{i\in N}P(R_{i})$, and the $\mathbf{L}$%
-minimum linear order $L$ of $A$ such that $L\supseteq \cap _{i\in N}P(R_{i})
$ otherwise. Clearly, neither of them satisfy the (strong) Pareto principle:
however, $f^{UN}$ satisfies BP and violates WP while $f^{\mathbf{L}x^{\ast }}
$ \ satisfies WP and violates BP. Nevertheless, WP has been widely used in
the extant literature, whereas BP has been rarely if ever explicitly
employed.
\end{remark}

\bigskip 

\begin{proposition}
There exist social welfare functions $f:\mathcal{R}_{A}^{N}\rightarrow 
\mathcal{R}_{A}$ which satisfy AMP$_{P}$, AN, ID, WNT, BP and are
strategy-proof on the domain $\mathcal{D}^{sp(d)}$ of single-peaked \ 
\textit{preorders }on $\dbigcup \limits_{B\in \mathcal{P}(A)}\mathcal{R}_{B}$%
\textit{\ }.
\end{proposition}

\begin{proof}
Since $(\dbigcup \limits_{B\in \mathcal{P}(A)}\mathcal{R}_{B},\subseteq)$ is
a median join-semilattices by Claim 2 above, it follows that Theorem 1 above
applies hence positive weakly-neutral quota social welfare functions as
defined above satisfy AN, ID, WNT and are strategy-proof on the domain $%
\mathcal{D}^{sp(d)}$ of single-peaked \  \textit{preorders }on $\dbigcup
\limits_{B\in \mathcal{P}(A)}\mathcal{R}_{B}$. Moreover, they also satisfy
AMP$_{P}$ by Claim 1, and BP by Claim 3, and the thesis is established.
\end{proof}

\bigskip

\begin{remark}
Observe that several domains of preference profiles are being considered
here. The first one consists of \ profiles $R_{N}=(R_{i})_{i\in N}$ of
arbitrary total preorders on the set $A$ of basic alternatives. The second
domain consist of profiles $\mathbf{R}_{N}=\mathbf{(R}_{i})_{i\in N}$ of \
single-peaked (partial) preorders on the ground set $\mathcal{R}_{A}$ of the
median join-semilattice of total preorders on $A$ (with single-peakedness
induced by the median betweenness of $\mathcal{R}_{A}$, and $\mathbf{R}_{N}$
induced by $R_{N}$). The third domain amounts to the subdomain of the second
one which only includes the metric single-peaked profiles $\widehat {\mathbf{%
R}}_{N}=\mathbf{(}\widehat{\mathbf{R}}_{i})_{i\in N}$ of preorders on $%
\mathcal{R}_{A}$ that are entirely determined by profiles $R_{N}$ through
their graphic distance from the top. Accordingly, we can consider different
notions of WP and BP: namely,

(1) WP (BP) of \ each `social preference' $f(R_{N})$ with respect to $R_{N}$%
:\ the requirement that, at any $R_{N}$, $f(R_{N})$ should faithfully
reflect any unanimous preference for an alternative in $A$ over another one;

(2) WP (BP) of \ each `social preference' $f(R_{N})$ with respect to $%
\mathbf{R}_{N}$ (or $\widehat{\mathbf{R}}_{N}$): the requirement that, at
any $R_{N}$, $f(R_{N})$ should be consistent with unanimous preferences
according to $\mathbf{R}_{N}$ (or $\widehat{\mathbf{R}}_{N}$), which means
that there should be no alternative `social preference' $R^{\prime}\in%
\mathcal{R}_{A}$ that is unanimously preferred over $f(R_{N})$ according to $%
\mathbf{R}_{N}$ (or $\widehat{\mathbf{R}}_{N}$).

Two most remarkable points are to be made here concerning the social welfare
functions mentioned in the previous Proposition. First, such social welfare
functions fail to satisfy WP with respect to the first and second domains,
consisting respectively of arbitrary profiles $R_{N}$ of total preorders on $%
A$, and of single-peaked profiles $\mathbf{R}_{N}$ of preorders on $\mathcal{%
R}_{A}$. Second, the very same social welfare functions do satisfy WP with
respect to the domain consisting of metric single-peaked profiles $\widehat{%
\mathbf{R}}_{N}$ of total preorders on $\mathcal{R}_{A}$.
\end{remark}

\bigskip

\begin{remark}
It is worth noticing that all of the anonymous, idempotent and
strategy-proof social welfare functions mentioned in the previous
proposition (including those which satisfy the Basic Pareto Principle BP
e.g. positive quota social welfare functions) admit a stalemate as one of
the possible outcomes, arising from certain specific patterns of strong
conflict among individual preferences. By definition, such (contingent)
stalemates give rise to violations of \ the Weak Pareto principle (WP) by
the chosen `social preference' both with respect to the `basic' preference
profiles $R_{N}$ of total preorders on the set $A$ of alternatives, and with
respect to general single-peaked domains $\mathbf{R}_{N}$ on $\mathcal{R}%
_{A} $ induced by the former $R_{N}$ profiles. Thus, the foregoing social
welfare functions may also be regarded as valuable sources of information
and advice concerning the `general interest' (or `common good'). In many
cases, they provide an explicit description of the alternatives that best
represent the `common good', or define anyway clear improvements on the
status quo. But occasionally they may also help to pursue the `general
interest' by pointing to situations of pathologically strong social
conflict: they do that precisely by returning outcomes that allow for
`inefficient' choices when fed with inputs encoding such a sort of social
conflict\footnote{%
See also Saari (2008) on the connection between conflict, cycling and
inefficiency.}. To put it in other terms, any violation of WP by such social
welfare functions might be regarded as a sort of `error message' calling for
public intervention (e.g. promoting an improved access to key relevant
information for the general public, implementing some appropriate
redistribution policies, or just relying on some contingent agenda
manipulation activities of the sort thoroughly analyzed and discussed in
Schwartz (1986)\footnote{%
Notice, however, that in IIA-based models of collective choice as advocated
by Schwartz (1986) such agenda-structure manipulation activities are treated
as \textit{normal} and endemic to every democratic aggregation protocol. Of
course, that is pretty much the same conclusion as that typically suggested
by authors that regard Arrow's theorem as an indictment of democratic
preference aggregation protocols, \`{a} la Riker (1982) (see footnote 26
below). By contrast, within the IIA-free models considered in the present
work, such agenda-structure manipulation processes can (and should) be
considered as \textit{local, contingent} subroutines appended to general
democratic aggregation protocols in order to increase their effectiveness to
cope with certain \textit{specific }sorts of conflicts related to Condorcet
cycles.} in order to ensure outcome-efficiency).
\end{remark}

\bigskip

Concerning the study of agenda manipulation-proofness for SAFE social
welfare functions and their agenda, Proposition 2 implies that IIA does in
fact enforce AMP$_{S}$. Unfortunately, the side effects of IIA on proper
consensus-based social welfare functions are simply devastating. That is
well-known thanks to a cluster of theorems originating with Arrow's famous
`general possibility theorem'. We recall here just a selected sample of four
key results from that cluster. The first pair consists of two
characterizations of \textit{dictatorial} social welfare functions: both of
them rely on the combination of IIA with some further condition which at
first sight would seem to be rather uncontroversial or at least undemanding.

\bigskip

\begin{theorem}
(i) (Arrow's Theorem (Arrow (1963))) A social welfare function $f:\mathcal{R}%
_{A}^{N}\rightarrow \mathcal{R}_{A}$ satisfies IIA and WP if and only if $f$
is dictatorial;

ii) (Hansson' s \textit{Non-Constancy }Theorem (Hansson (1973))\footnote{%
It should be emphasized that Hansson's Non-Constancy Theorem (which was
established independently of Wilson's Theorem) amounts to replacing the
`global stalemate' clause of Wilson's Theorem \ with a weaker clause
(violation of the LS condition i.e. of `Strong Non-Constancy' in Hansson's
own original terminology). Incidentally, a close inspection of Hansson's
proof shows that it can also be deployed to imply the stronger Wilson's
`global stalemate' clause. See also Malawski, Zhou (1994) and Cato (2012)
for related work on preference aggregation without WP.} ) A social welfare
function $f:\mathcal{R}_{A}^{N}\rightarrow \mathcal{R}_{A}$ satisfies IIA,
LS and is \textbf{not }inversely dictatorial if and only if $f$ is
dictatorial.\textit{\ }
\end{theorem}

\bigskip

The third theorem focusses instead on the combination of IIA with a
definitely compelling and uncontroversial condition, to point out the
exceedingly strong and unpalatable restrictions that combination engenders
on the admissible social welfare functions.

\bigskip

\begin{theorem}
(Wilson (1972)) Let $f:\mathcal{R}_{A}^{N}\rightarrow \mathcal{R}_{A}$ be a
social welfare function that satisfies IIA and WS. Then, $f$ is either
dictatorial, or inversely dictatorial, or else $f=f^{U_{A}}$ i.e. $f$ is the
global stalemate social welfare function for $(N,A)$.
\end{theorem}

\begin{remark}
Notice that the original proof of Hansson's Non-Constancy Theorem in Hansson
(1973) relies in fact on Arrow's Theorem. Moreover, a careful inspection of
that proof makes it clear that the only social welfare function that is
neither dictatorial nor inversely dictatorial and satisfies IIA is the
Global Stalemate function. In other terms, Hansson's proof shows that
Arrow's Theorem implies Wilson's Theorem (it should be recalled here that
Hansson's Non-Constancy Theorem was first published in a 1972 working paper,
independently of Wilson's Theorem). But then, since Wilson's Theorem
obviously implies Arrow's Theorem, the foregoing observation confirms that
the two of them are in fact equivalent.
\end{remark}

\bigskip

The fourth theorem shows that combining IIA with two widely accepted
conditions for `democratic' aggregation rules such as anonymity (AN) and
neutrality (NT) results in a characterization of the global stalemate social
welfare function.

\bigskip

\begin{theorem}
(Hansson (1969a)) Let $f:\mathcal{R}_{A}^{N}\rightarrow \mathcal{R}_{A}$ be
a social welfare function that satisfies IIA, AN and NT. Then, $f=f^{U_{A}}$
i.e. $f$ is the global stalemate social welfare function for $(N,A)$.
\end{theorem}

\bigskip

The following alternative characterization of the global stalemate social
welfare function combines IIA with two weak conditions following from
anonymity and neutrality such as WS and MDR to the effect of emphasizing the
inordinate strength of IIA.

\bigskip

\begin{proposition}
A social welfare function $f:\mathcal{R}_{A}^{N}\rightarrow \mathcal{R}_{A}$
satisfies IIA, WS and MDR if and only if $f$ is the global stalemate social
welfare function i.e. $f=f^{U_{A}}$.
\end{proposition}

\begin{proof}
$\Longrightarrow$ Suppose that social welfare function $f$ satisfies IIA, WS
and MDR. But then, it follows from Wilson's Theorem as mentioned above that $%
f$ is dictatorial, inversely dictatorial, or the global stalemate constant
function $f^{U_{A}}$. However, a dictatorial social welfare function clearly
violates MDR: indeed, suppose $i\in N$ is such that, for every $x,y\in A$
and $R_{N}\in$ $\mathcal{R}_{A}^{N}$, $xf(R_{N})y$ entails $xR_{i}y$.
Moreover, by WS, for every $x,y\in A$ there exists $R_{N}\in \mathcal{R}%
_{A}^{N}$ such that $xf(R_{N})y$, whence $xR_{i}y$ holds. Then, by MDR there
exists $S\subseteq N\smallsetminus \left \{ i\right \} $ and a pair of
distinct $v,z\in A$ such that $S\in F_{f}(v,z)$ i.e. there exists $%
R_{S}^{vz}\in \mathcal{R}_{A}^{S}$ such that $vf(R_{N})z$ for every $R_{N}\in 
\mathcal{R}_{A}^{N}$ with $R_{S|\left \{ v,z\right \} }=R_{S|\left \{
v,z\right \} }^{vz}$. Now, consider a profile $R_{N}\in \mathcal{R}_{A}^{N}$
such that $zR_{i}v$, \textit{not }$vR_{i}z$ and $R_{S|\left \{ v,z\right \}
}=R_{S|\left \{ v,z\right \} }^{vz}$. By definition of $F_{f}$, $vf(R_{N})z$%
. However, since $f$ is dictatorial, \textit{not }$vR_{i}z$ implies \textit{%
not }$vf(R_{N})z$, a contradiction. Thus, $f$ is \textit{not} dictatorial,
as required.

Similarly, suppose there exists $i\in N$ such that for every $x,y\in A$ and $%
R_{N}\in$ $\mathcal{R}_{A}^{N}$, $xf(R_{N})y$ entails $yR_{i}x$. Again, it
follows from WS that for every $x,y\in A$ there exists $R_{N}\in$ $\mathcal{R%
}_{A}^{N}$ such that $xf(R_{N})y$, whence $yR_{i}x$ holds, by our
assumption. Then, by MDR there exists $S\subseteq N\smallsetminus \left \{
i\right \} $ and a pair of distinct $v,z\in A$ such that $S\in F_{f}(v,z)$.
Now, consider a profile $R_{N}\in \mathcal{R}_{A}^{N}$ such that $vR_{i}z$, 
\textit{not }$zR_{i}v$ and $R_{S|\left \{ v,z\right \} }=R_{S|\left \{
v,z\right \} }^{vz}$. By definition of $F_{f}$, $vf(R_{N})z$. However, by
assumption, \textit{not }$zR_{i}v$ implies \textit{not }$vf(R_{N})z$, a
contradiction. Thus, $f$ is \textit{not} inversely dictatorial either.
Therefore, it follows from Wilson's Theorem that $f=f^{U_{A}}$, the global
stalemate function.

$\Longleftarrow$ It can be easily shown that the global stalemate social
welfare function $f^{U_{A}}$ satisfies IIA, WS and MDR. Indeed, take any
sovereign agenda formation rule $f$, posit $\mathcal{F}\mathbb{(}f):=\left
\{ f_{B}:=f^{U_{B}}\right \} _{B\subseteq A}$, and consider the
corresponding PAFE social welfare function $\mathbf{f}=(f,\mathcal{F}\mathbb{%
(}f))$. By definition, $f_{A}:=f^{U_{A}}$ which obviously satisfies IIA,
being a constant function defined on $\mathcal{R}_{A}^{N}$. Thus, by
Proposition 2, $\mathbf{f}$ satisfies AMP$_{S}$ and consequently, by
definition, $f$ $^{U_{A}}$ also satisfies AMP$_{S}$. Moreover, for any $%
x,y\in A$ and $R_{N}\in \mathcal{R}_{A}^{N}$, $xf^{U_{A}}(R_{N})y$ hence WS
is trivially satisfied by $f^{U_{A}}$. Finally, observe that the
responsiveness correspondence $F_{f^{U_{A}}}$ is such that $%
F_{f^{U_{A}}}(x,y)=\mathcal{P}(N)$ for all $x,y\in A$. But then, for any $%
x,y\in A$ and any $i,j\in N$, $\left \{ N,\left \{ i\right \} ,\left \{
j\right \} \right \} \subseteq F_{f^{U_{A}}}(x,y)$. It follows that $%
f^{U_{A}}$ also satisfies MDR.
\end{proof}

\bigskip

\begin{remark}
Notice that AN and NT do indeed imply WS and MDR, while the converse does
not hold: to see this, consider the social welfare function $f^{\ast}$ such
that for some $1,2,3\in N$, and for every $x,y\in$ $A$, $R_{N}\in$ $\mathcal{%
R}_{A}^{N}$, $xf^{\ast}(R_{N})y$ iff either $xR_{\left \{ 1,2\right \} }y$
or [\textit{not }$xR_{\left \{ 1,2\right \} }y$ and $xR_{3}y $]. It follows
that Proposition 4 amounts to an extension of Hansson's characterization of
the Global Stalemate social welfare function $f^{U_{A}}$via AN, NT and IIA
(Hansson(1969a)).
\end{remark}

\bigskip

\begin{corollary}
There is no social welfare function $f:\mathcal{R}_{A}^{N}\rightarrow 
\mathcal{R}_{A}^{T}$ that satisfies IIA, S and MDR. Thus, in particular,
there is no idempotent social welfare function that satisfies IIA and MDR.
\end{corollary}

\begin{proof}
Suppose that on the contrary there exists a social welfare function $f$
which satisfies IIA, S, and MDR. Since S clearly implies WS, it follows from
Proposition 4 above that $f=f^{U_{A}}$, a contradiction because by
definition $f^{U_{A}}$does \textit{not} satisfy S. The second statements
follows trivially since any idempotent social welfare function does satisfy
S.
\end{proof}

\bigskip

So, we have a further impossibility result that follows just from the
combination of IIA and MDR with Sovereignty, with no role at all for the
Weak Pareto Principle.

It should also be mentioned that, relying on other well-known results from
the extant literature, further elaborations on the role of IIA established
by the foregoing theorems can be easily produced. For instance, a further
result in Wilson (1972) shows that \textit{any} social welfare function
which satisfies IIA must produce `social preferences' invariably composed by 
\textit{some} combination of \textit{at most five }different types of
patches corresponding respectively to `locally imposed strict preferences',
`minimal (local) stalemates', `non-minimal (local) stalemates', `locally
dictatorial preferences' and `locally inversely-dictatorial preferences'%
\footnote{%
See Wilson (1972), Theorem 5. Binmore (1976) is an interesting further
extension of that theorem, showing that in order to avoid the grim
consequences of the latter the domain of a social welfare function which
satisfies IIA should be dramatically restricted. Specifically, its domain
should not include \textit{all} the preference profiles of total preorders
consistent with \textit{at least one} arbitrarily fixed `party structure'
(namely, a partition of agents into `parties' defined as sets of agents
whose preferences are al least partially concordant on \textit{every} pair
of alternatives).}. Furthermore, it is also well-known that there exists a
quite general model-theoretic rationale underlying such results (see e.g.
Lauwers, Van Liedekerke (1995) for details). Namely, it is sufficient to
join \textit{either} IIA and the Weak Pareto Principle (WP) \textit{or }IIA
and Idempotence (ID) to force the set of \textit{decisive coalitions%
\footnote{%
A \textit{decisive coalition} of $\ $a social welfare function $f$ \ is any
coalition $C\subseteq N$ that can enforce the unanimous preference of its
members between each pair of alternatives as the actual social preference. } 
}of a social welfare function $f$ for $(N,A)$ to be an \textit{ultrafilter%
\footnote{%
An \textit{ultrafilter} (or \textit{maximal lattice-filter} on $N$) is a
nonempty set $\mathcal{F\subseteq P}(N)\setminus \left \{ \varnothing \right
\} $ such that for every $C,D\in \mathcal{F}$: (i) $C\cap D\in \mathcal{F}$
\ and (ii) either $C\in \mathcal{F}$ or $N\setminus C\in$ $\mathcal{F}$.
Since $N$ is finite, every ultrafilter $\mathcal{F}$ on $N$ is \textit{%
principal}, namely $\mathcal{F}=\left \{ C\subseteq N:i\in C\right \} $ for
some $i\in N$. It follows that $\left \{ i\right \} $ is a decisive
coalition for $f$, and consequently $f$ is a \textit{dictatorial }social
welfare function.}} on $N$, a fact which in turn implies that $f$ is \textit{%
dictatorial }since $N$ is by assumption finite.

Thus, our results confirm Arrow's core intuition concerning the relationship
of IIA to agenda manipulation-proofness. At the same time, they also
circumscribe that relationship to the case of an `agenda-first' \textit{%
sequential coupling of agenda formation and preference aggregation }(the AMP$%
_{S}$ version of agenda manipulation-proofness fr SAFE social welfare
functions)\textit{. }Specifically, Proposition 4 shows that IIA both ensures
AMP$_{S}$ for any SAFE social welfare function, and implies a condition
(IIAP) that is \textit{required} to secure AMP$_{S}$. It follows that
Arrow's `general possibility theorem' can be regarded as the seminal result
of a cluster of theorems showing that reliance on IIA to achieve agenda
manipulation-proofness in the AMP$_{S}$ version for SAFE social welfare
function has an exhorbitant, unacceptable cost. To be sure, it remains to be
seen whether relaxing IIA to IIAP makes it possible to achieve AMP$_{S}$
without ruling out at all proper consensus-based social welfare functions,
and if so to what extent. But in any case, such a cluster of theorems (as
combined with Propositions 1 and 3 above) point to \textit{PAFE social
welfare functions as a feasible working alternative to secure the
possibility of reliable and proper consensus-based preference aggregation
rules}.

\bigskip

\textbf{3. Discussion.}

Thus, it is abundantly clear that IIA is a powerful obstruction\textit{\ }to 
\textit{each one }of the following two basic requirements for any proper
consensus-based social welfare function, namely (i) \textit{weak
Pareto-optimality} \textit{and }(ii) \textit{minimally distributed
responsiveness}. But the main motivation for introducing IIA is \textit{%
precisely} the attempt to prevent agenda manipulation of a social welfare
function when the elicited preferences only concern alternatives of the
agenda. Hence, it is insisting on IIA to ensure agenda
manipulation-proofness (as expressed by AMP$_{S}$) that gives rise \textit{%
by itself\ }to a bleak scenario concerning the construction of properly
consensual social welfare functions.

Let us then briefly summarize the widely shared views on the import of
Arrow's theorem which typically follow from a firm endorsement of IIA as
grounded on the assumption that IIA amounts to a key requirement for \textit{%
any }reasonable voting rule in order to prevent agenda manipulation. Since
Arrow's IIA is a property shared by several commonly used\ voting rules
including the simple majority rule (arguably, a paragon of `democratic'
voting rules), it seems to follow that Arrow's theorem does validate the
following challenging, momentous statement. Namely, the assertion that 
\textit{any attempt} to use \textit{democratic} voting rules to articulate a 
\textit{consistent formulation of the collective interest} with a view to
identify and select policies which best promote it \textit{is doomed to
failure}. That is so precisely because Arrow's theorem shows that under IIA (%
\textit{Weak)\ Pareto optimality can only be achieved through dictatorship}.
Therefore, since dictatorship is obviously to be rejected as a means to
define proper consensus-based `social preferences', insisting to prevent
agenda manipulation (specifically, agenda-content manipulation) entails
reliance on some aggregation rule which might license `social preferences'
that occasionally reverse unanimously held individual strict preferences
between some pairs of alternatives. And that is also deemed to be not
acceptable. But then, the only alternative left is \textit{to allow for
agenda manipulation} (specifically, agenda-content manipulation) to the
effect of undermining \textit{reliability} of the aggregation rule, since
the representation of the `general interest' provided by such a rule will
typically reflect just successful manipulatory activities, and possibly
nothing else. Either way, the aim of producing a consistent, faithful and
credible formulation of the `general interest' cannot be apparently
fulfilled. To put it bluntly, under that view majority voting cannot be
relied upon to discover the public interest or `general will' because of its
possible cycles, and \textit{nothing else can work }because of the same
vulnerability to agenda manipulation. As a consequence, in actual practice 
\textit{there is apparently not such a thing as `general interest' to
discover, express and implement as a guide or benchmark for public policy}
(see e.g. Schwartz (1970)). It also follows, accordingly, that there is
apparently no way to feed work in mechanism design and/or institutional
design with well-grounded and reliable criteria summarizing the `general
interest', aimed at improving the effectiveness of democratic institutions.%
\footnote{%
Riker (1982) is probably the most outspoken and consistent presentation of
such a view available in print. While Dahl (1956) tentatively labeled
`populist democracy' the doctrine that identifies the exercise of
sovereignty with \textit{exclusive and unrestrained reliance on majority
voting}, Riker assumes that Arrow's theorem licences the imputation of the
same, allegedly hopeless, limitations of the majority rule to \textit{every
possible `democratic' preference aggregation rule}. Consequently, virtually
all the social-choice-theoretic work in mechanism design is uncerimoniously
identified with `populism' and dismissed as a hopeless endeavour. Riker's
suggested and clearly preferred alternative is acceptance of the
`Schumpeterian' view that modern `democratic politics' is -and \textit{has
to be}- nothing else than (i) \textit{competitive electoral selection of the
ruling elite}, and (ii) pervasive and relentless activities of \textit{%
agenda manipulation }on the part of elected officials and representatives,
in view of more or less special interests, with no effective role left for a
public representation of the common interest as a shared, consensus-based
benchmark. Accordingly, analyzing `democratic politics' from such a
perspective is regarded as the main task of `democratic theory' proper, and
the plain endorsement of that view is the defining feature of
`liberalism'(see Riker (1982), and Schofield (1985) for a thorough technical
treatment of social choice theoretic models in a multidimensional spatial
setting that shares at least some of Riker's views, while refraining from
the latter's most ideological overtones).}

That scenario has been variously described as the impossibility of
`rational' collective decisions\footnote{%
See e.g. Buchanan (1954) for an early example of that view, insisting on the
alleged impossibility of collective decisions replicating the `rationality'
of individual decisions. A much more elaborated argument to a similar effect
is offered by Bordes, Tideman (1991), showing that IIA is a consequence of a
strong restriction on general voting rules called `regularity' that is
however scarcely compelling \textit{unless systematic irredundancy} \textit{%
of preference elicitation} is -again- tacitly assumed. A somewhat more
balanced view, advocating a combination of IIA with less demanding criteria
of `rationality' for \textit{both} collective and individual decisions is
advanced by Schwartz (1986).}, or the impossibility of a reliable and
significant consensus-based expression of the `collective good' (or `general
interest' or `public will'\footnote{%
See e.g. the highly influential Riker (1982) as discussed above (footnote
25).}). This is in fact the most familiar understanding of Arrow's
`impossibility theorem'. Such an interpretation projects a `dark' view of
the content and significance of Arrow's theorem concerning the viability and
effectiveness of democratic protocols since it suggests, in short, that
there is no way to use voting methods, decision systems or preference
aggregation rules of any sort to help improving the effectiveness and
deliberative quality of current democratic protocols\footnote{%
An influential tentative list of basic, \textit{substantive} requirements
for democratic decision protocols including `political equality',
`deliberation', `participation', and `agenda control' is due to Dahl (1979).}%
. Notice however that, again, all of the above rests crucially on the
working assumption that IIA is both reasonable and virtually inescapable%
\footnote{%
See for instance Schwartz (1986), p.33 or Bordes,Tideman (1991)) for a
remarkably clear, adamant endorsement of that view.}.

But then, Propositions 1 and 3 show that \textit{there exists in fact an
alternative way to achieve agenda manipulation-proofness via AMP}$_{P}$:
such an alternative makes it possible to devise anonymous and idempotent
social welfare functions that satisfy a basic version of the Pareto
principle, and are -in a compelling sense- \textit{both agenda
manipulation-proof and strategy-proof} \textit{. }Thus, there is indeed an
effective way out of the strictures identified by Arrow's theorem. From that
perspective, Arrow's result actually provides \textit{constructive}
information about the design of social welfare functions and preference
aggregation rules: in that sense, there is also a \textit{bright side }of
Arrow's theorem. Access to the latter requires three basic steps:

(i) reliance on (possibly \textit{redundant}) \textit{preference elicitation
concerning an entire set of prefixed admissible alternatives\ }in order to
ensure \textit{agenda manipulation-proofness} \textit{without any recourse
to IIA};

(ii) a \textit{mild relaxation of the Pareto Principle to BP allowing for
occasional stalemates} (namely, social indifference over a set of
alternatives including Pareto-dominated outcomes), and a concurrent
reinterpretation of possible violations of WP and other, stronger, versions
of the Pareto principle as `warning signals' pointing to the need for
remedial actions including policies to correct blatant disparities of access
to information and/or other key resources;

(iii) refocussing on \textit{a further condition (`monotonic }$M_{\mathcal{X}%
}$\textit{-independence')} in order to address strategic manipulation
issues: such a condition can be regarded as a combination of a mild
monotonicity property and a considerably \textit{weakened} version of IIA.

Broadly speaking, some form of each one of the foregoing steps was
previously considered or at least evoked in the extant literature (a
discussion of that matter is provided in the next section).\footnote{%
Remarkably, Przeworski (2011) suggests both relaxations of IIA and some
judicious use of agenda-structure manipulation as possible remedies to cope
with possible violations of WP, and secure a reliable representation of the
collective interest based on `proximity' to individual preferences. However,
his somewhat informal treatment leaves it unexplained whether or not he also
proposes to combine such remedial moves (and if so, how).} What is new here
is, arguably, \textit{their joint consideration as made possible by a model
that combines agenda formation and preference aggregation}. The resulting
analysis shows that a sound \textit{parallel-coupling }of (`full domain')
social welfare functions to their own agenda-formation processes makes it
possible to jointly achieve \textit{anonymity}, respect for \textit{unanimity%
}, a \textit{basic form of Pareto optimality}, and \textit{agenda
manipulation-proofness} together with \textit{strategy-proofness} on a
suitably large and very natural domain of single-peaked
`meta-preferences'(or `preferences on preferences'). As observed above, it
is also remarkable that the latter strategy-proofness property turns out to
be essentially equivalent to an `independence' condition which amounts to a 
\textit{considerable weakening of IIA.}

\section{Related work}

As mentioned in the introduction, agenda manipulation-proofness was used by
Arrow as the main motivation for introducing IIA as a basic requirement for
social welfare functions in his seminal work (Arrow (1963)). Since then, it
has become quite common to use `Arrowian' as a qualifier for social welfare
functions or aggregation rules which satisfy some version of `independence
of irrelevant alternatives' (and possibly some further basic requirement
such as idempotence i.e. `respect for unanimity')\footnote{%
See, among many others, Aleskerov (1999), Sethuraman, Teo, Vohra (2003),
Nehring, Puppe (2010). An alternative (and perhaps more appropriate) usage
is the one that rather contrasts Arrowian (or multi-profile) and
Bergson-Samuelson (or single-profile) social welfare functions, and goes as
follows. Let $N,A$ \ be two (finite) sets and $\mathcal{R}_{A}$ the set of
all total preorders on $A$. An \textit{Arrowian social welfare function} for 
$(N,A)$ is a function
\par
$f:$ $\mathcal{R}_{A}^{N}\rightarrow \mathcal{R}_{A}$,
\par
while a \textit{Bergson-Samuelson social welfare function }for $(N,A)$ is a
function
\par
$f:\left \{ r^{N}\right \} \rightarrow \mathcal{R}_{A}$, with $r^{N}\in 
\mathcal{R}_{A}^{N}$.
\par
Moreover, their \textit{strict }counterparts are obtained by replacing $%
\mathcal{R}_{A}$ with\textit{\ }$\mathcal{L}_{A}$ i.e. the set of all 
\textit{linear orders} (namely, \textit{antisymmetric }total preorders) on $%
A $.
\par
{}}. Hence, the amount of literature which is broadly related to the topics
covered by the present paper is simply enormous. Therefore, we shall confine
ourselves to a very brief comment focussing on the contributions that raise
points most strictly related to the content of the present work (the
interested reader is addressed to the supplementary Appendix for a more
extensive and detailed discussion).

Monjardet (1990) first introduces the notion of (\textit{monotonic) }$M_{%
\mathcal{X}}$\textit{-independence }presenting it as a generalized
counterpart of IIA, but is not concerned with nonmanipulability properties
of any sort. Dietrich, List (2007b) shows the equivalence between a similar
notion of monotonic independence as specialized to propositions of a certain
class of formal languages and strategy-proofness of \textit{judgment}
aggregation rules. Savaglio, Vannucci (2019) shows the equivalence of 
\textit{monotonic }$M_{\mathcal{X}}$\textit{-independence }and
strategy-proofness for arbitrary aggregation rules in bounded distributive
lattices, and Vannucci (2019) applies that equivalence to \textit{rating }%
aggregation rules. Nehring, Puppe (2007, 2010) rely on (finite) property
spaces to define a class of generalized IIA properties by reducing each one
of them to a specific family of binary issues. Thus, IIA and $M_{\mathcal{X}%
} $-independence\textit{\ }are easily obtained as special cases of that
construct (but the foregoing papers do not consider the latter property, and
are not concerned with either agenda manipulation-proofness or
strategy-proofness of social welfare functions). Neither agenda
manipulation-proofness nor strategy-proofness properties are considered in
Huang (2014). However, that work provides an escape route from Arrow's
`general possibility theorem' relying on the combination of a newly defined
weakened version of IIA and the same basic Pareto optimality condition
employed in the present work, to the effect of allowing \textit{stalemates }%
(or `singularities' in Huang's own terminology). Sato (2015) is partly
concerned with the relationship between agenda manipulation-proofness and
IIA, but only considers a special class of \textit{strict} social welfare
functions enjoying a certain continuity property (`bounded response').
Curiously enough, the preference-approval aggregators studied in Kruger,
Sanver (2021), that are \textit{not} meant to address agenda manipulation
issues but rather an entirely different problem (the joint aggregation of
rankings and binary ratings), are in fact isomorphic to a \textit{%
restricted-domain} version of our PAFE social welfare functions: such a
domain-restriction is indeed necessary to capture the required consistency
between individual rankings and ratings. Finally, it should also be noticed
that the replacement of IIA with (monotonic) $M_{\mathcal{X}}$-independence
that is suggested in the present work amounts in general to a \textit{%
remarkable change and enlargement of the information-base} that is made
available for the definition of the relevant social welfare function. That
is so because when expressed in terms of properties of $\mathcal{R}_{A}$
with $|A|=m\geq 3$, the size of the set of actually relevant \textit{binary
issues} (for each individual preference relation of any profile) is $m(m-1)$
for IIA, and $2(2^{m-1}-1)$ for $M_{\mathcal{X}}$-Independence. That move is
of course consonant with a long standing argument in the social
choice-theoretic literature, that counts Sen (2017) and Saari (2008) among
its most committed and distinguished advocates. Yet, the present enlargement
rests on a peculiar combination of two characteristic features that single
it out from many other proposed enrichments of the information-base.\
Namely, it relies on a \textit{standard input-format consisting of profiles
of total preorders}, and at the same time \textit{exploits the structure of
the set of total preorders} themselves i.e. the objects to be aggregated.

\section{Concluding remarks}

The IIA condition for preference aggregation rules was first introduced by
Arrow in order to ensure their agenda manipulation-proofness, but when
combined with a few minimal reasonable conditions it results in a
characterization of dictatorial social welfare functions. That is the basic
content of Arrow's general possibility theorem. Under the assumption that
IIA is indeed the only way to block agenda manipulation, that theorem does
also imply that reliable and proper consensus-based social welfare functions
do not exist. Now, that interpretation projects a negative, disagreeable
shadow on the perceived consequences of Arrow's theorem since it suggests
that\  \textit{no meaningful consensus-based formulation of `general
interest' is available} as a guide and benchmark to promote and assess
public decisions and policies, and improve the design and/or implementation
of democratic protocols. Thus, it is not unfair to describe all that as the `%
\textit{dark}' side of Arrow's theorem. The present work shows however that,
as a matter of fact, agenda manipulation-proofness of a social welfare
function is indeed available \textit{without any appeal to IIA }provided
that agenda formation and preference elicitation are \textit{simultaneous}.
In the latter case some \textit{anonymous, idempotent, agenda
manipulation-proof, minimally efficient and weakly-neutral social welfare
functions do exist}. Moreover, a \textit{much relaxed independence condition 
}that they do satisfy ensures their \textit{strategy-proofness} as well. But
then, from the perspective provided by such \textit{positive} results,
Arrow's theorem may also be regarded as a most \textit{constructive}
contribution to the design of preference aggregation rules, in that it
suggests that agenda formation and preference elicitation are better \textit{%
not }coupled sequentially. That is precisely the `\textit{bright}' side of
Arrow's theorem that the present work is meant to highlight and emphasize.

To be sure, the consensus-based, agenda manipulation-proof, and
strategy-proof preference aggregation rules we have shown to be available
require a significant increase of the amount of information to be extracted
from preference profiles, and processed. Thus, reliance on such aggregation
rules also involves a careful consideration of \textit{computational
complexity} issues, and their implications\footnote{%
Indeed, it is well-known that the computation of median total preorders for
arbitrary profiles of total preorders is NP-complete (i.e. it belongs to the
class of the hardest problems whose solutions are polynomial-time verifiable
or `easy' to verify, but apparently worst-case `hard' to compute).
Specifically, if the size of $N$ is suitably larger than the size of $A$,
computing a median total preorder is NP-complete for arbitrary profiles of
total preorders or linear orders (Hudry (2012)) and NP-hard (i.e. `easy' to
reduce to a NP-complete problem) for arbitrary profiles of binary relations
(Wakabayashi (1998)).}. Moreover, while \textit{individual}
strategy-proofness issues concerning social welfare functions have been also
considered in the present work, the further problems arising from \textit{%
coalitional} \textit{strategy-proofness }requirements for preference
aggregation rules have been deliberately put aside.\footnote{%
Concerning the relationships between individual and coalitional
strategy-proofness for general aggregation rules see e.g. Vannucci (2016).
\par
\bigskip
\par
\bigskip
\par
\bigskip
\par
{}} Finally, since we have shown that IIAP only (a weaker version of IIA) is
actually \textit{necessary} to ensure `sequential' agenda
manipulation-proofness in the required sense, it remains to be seen whether
further interesting aggregation rules are consistent with IIAP. Such most
significant and intriguing issues are however best left as challenging
topics for future research.

\bigskip \pagebreak

\section{Appendix}

The present supplementary Appendix provides a rather detailed discussion of
some related previous contributions collecting them under two distinct
subsections that correspond to two focal points of the present analysis both
related to IIA, namely `Agenda manipulation-proofness and IIA' and
`Strategy-proofness and weakenings of IIA'.

(I) \textit{Agenda manipulation-proofness and IIA.}

The role of preference elicitation on the \textit{entire} set of admissible
alternatives in order to ensure transitivity properties of the `social
preferences' and the resulting violation of IIA has been repeatedly pointed
out, and contrasted with \textit{peacemeal} elicitation of preferences on
specific agendas of admissible alternatives, which is conducive to
IIA-consistency and violation of transitivity properties of `social
preferences'(see e.g. Sen (1977) where that contrast is discussed with
reference to several versions of the simple majority rule and the Borda
Count scoring rule). Unfortunately, preference elicitation and agenda
formation are typically \textit{not} modelled together in the extant
literature: specifically, agenda manipulation is usually left unmodelled and
thus given a quite informal treatment\footnote{%
A partial exception due to Dietrich (2016) is available in the related
framework of \textit{judgment aggregation }(to be discussed below) where
agenda manipulation is modeled as sensitivity of aggregate judgments on
issues to agenda-content alterations (\textit{including expansions}), with
no explicit role for preferences. Notably, that notion of agenda
manipulation-proofness is shown to be tightly connected to a version of IIA,
and results in a characterization of dictatorial judgment aggregation rules
when combined with a unanimity-respecting condition for general finite
agendas.}. As a consequence, even the agenda-content and agenda-structure
dimensions of agenda manipulation are typically \textit{not }neatly and
consistently distinguished. Thus, following the lead of a much discussed,
misleading example proposed in Arrow (1963), use of the label `IIA' has also
been occasionally stretched to also refer (improperly) to requirements on
social preference rankings on a certain subset of alternatives across
several \textit{distinct} admissible agendas for a \textit{fixed} profile of
individual preferences on the largest admissible agenda (see e.g. Ray
(1973), Fishburn (1973), Sen (1977), Schwartz (1986), Bordes, Tideman
(1991), Young (1995)\footnote{%
Actually, Young (1995) also discusses at length a condition he calls `local
stability' or \textit{local independence of irrelevant alternatives (LIIA) }%
which is also satisfied\textit{\ by median-based aggregation rules }(but not
by positionalist rules such as the Borda Count)\textit{. }When applied to a
social welfare function $f:\mathcal{R}_{A}^{N}\rightarrow \mathcal{R}_{A}$
LIIA may be formulated as follows:
\par
$f((R_{N})_{|B})=(f(R_{N}))_{|B}$ for all $R_{N}\in \mathcal{R}_{A}$, and
all $B\in \mathcal{I}_{f(R_{N})}$ where, for any $R\in \mathcal{R}_{A}$,
\par
$\mathcal{I}_{R}:=\left \{ 
\begin{array}{c}
I\subseteq A:I=\left \{ z:xRzRy\right \} \text{ } \\ 
\text{for some }x,y\in A%
\end{array}
\right \} $.
\par
Thus LIIA is in fact an \textit{intraprofile }property\textit{, }rather than
an \textit{interprofile }property like IIA and its relaxed versions (again,
see Fishburn (1973) for a classic, exhaustive classification of standard
social choice-theoretic properties for preference aggregation rules).} for
discussions related to such topic). Furthermore, without an explicit joint
modelling of agenda formation and preference elicitation it is virtually
impossible to distinguish not only between parallel coupling and sequential
coupling of those two processes, but also between preference-first and
agenda-first sequential coupling. In the previous section of the present
work it has been shown that parallel-coupling allows for agenda
manipulation-proofness of social welfare functions without any recourse to
IIA, while under agenda-first sequential coupling agenda
manipulation-proofness is in fact \textit{strictly related to IIA. }But
then, what about \textit{preference-first sequential coupling }of agenda
formation and preference elicitation if outputs are just `social choice
sets' out of the entire admissible set that might be not representable as
optima of an underlying total preorder of social preferences? Under such
circumstances, the possibility of agenda-structure manipulation re-enters
the picture. In that connection, the main theorem of Hansson (1969b)
concerning \textit{generalized social choice correspondences (GSCCs)}%
\footnote{%
Or `group decision functions' in the original terminology of Hansson (1969b).
\par
A \textit{generalized social choice correspondence }for $(N,A)$ is a function
\par
$f:\mathcal{R}_{A}^{N}\longrightarrow \mathcal{C}_{\mathcal{A}}$ where $%
\mathcal{A}\subseteq \mathcal{P}(A)\setminus \left \{ \emptyset \right \} $
with $A\in \mathcal{A}$ and $\mathcal{C}_{\mathcal{A}}$ is the set of all
functions $C:\mathcal{A}\longrightarrow \mathcal{P}(A)\setminus \left \{
\emptyset \right \} $ such that $C(B)\subseteq B$ for every $B\in \mathcal{A}
$.
\par
A \textit{social choice correspondence} for $(N,A)$ is a function $f:%
\mathcal{R}_{A}^{N}\longrightarrow \mathcal{P}(A)\setminus \left \{
\emptyset \right \} $ i.e. a generalized social choice correspondence such
that $\mathcal{A}=\left \{ A\right \} $.
\par
A social choice correspondence for $(N,A)$ whose range consists of \textit{%
singleton-sets }is also said to be a \textit{social choice function, }and
usually written $f:\mathcal{R}_{A}^{N}\longrightarrow A$.
\par
{}}, and its reformulation and extension due to Denicol\`{o} (2000) are
indeed relevant and most helpful. In fact, the foregoing Hansson's result
relies on an \textit{extended }version of IIA for GSCCs that are not
necessarily generated through maximization of the total preorders that
express social preferences. Specifically, it implies that any \textit{social
choice correspondence }$F$ on a set $A$ (with $|A|\geq3$) that satisfies WP
and such an extended IIA property can be represented as the choice of maxima
of the social preferences in the range of a social welfare function $f$
which satisfies IIA and WP if and only if both $F$ and $f$ are dictatorial
(see Hansson (1969b), Theorem 3 \footnote{%
See also Denicol\`{o} (2000) for a simplified presentation of Hansson's
theorem, and a detailed formulation of its consequences for social choice
correspondences and social welfare functions as just mentioned in the text.}%
). In a similar vein, Ferejohn, McKelvey (1983) shows that even substituting
transitivity and totality of social preferences with social choice sets that
are Von Neumann-Morgenstern solutions of an asymmetric social
dominance/preference relation, insistence on IIA (in a slightly strenghtened 
\textit{monotonic} version that implies WP, actually) results in the
alternative between allowance for at least one agent with unlimited veto
power and allowance for \textit{agenda-structure manipulation} as defined in
the Introduction.

\bigskip

(II) \textit{`Strategy-proofness and weakenings of IIA'. }The other major
theme in the present work is that, once agenda manipulation-proofness of
properly consensus-based social welfare function is secured through
parallel-coupling of agenda formation and preference elicitation (with no
role at all for IIA), the strategy-proofness issue for such social welfare
functions does also admit a sensible formulation and a positive solution.
Specifically, the latter requires just (a) focussing on the `right'
individual preferences (which \textit{must} be preferences on the outcomes
of a social welfare function, hence \textit{preferences on social
preferences over outcomes} i.e. ultimately `meta-preferences on basic
preferences') and (b) observing that \textit{basic preferences on
alternatives induce in a natural way single-peaked `meta-preferences' on the
`preference space' which in turn ensure strategy-proofness of the proper
consensus-based social welfare functions }mentioned above. Moreover, it
turns out that in such a setting strategy-proofness is in fact equivalent to
the combination of a very mild \textit{monotonicity} condition on the
influence of coalitions (namely the requirement that adding support to a
previously positive decision on a certain binary issue should never result
in a decision reversal) and an \textit{independence }condition that amounts
to a \textit{much weakened version of IIA. }

Thus, in a sense, a certain version of IIA ultimately reenters the picture
but (i) in a \textit{much weakened} and \textit{very specific }form\ 
\footnote{%
To be sure, it is also well-known that IIA is \textit{so strong }that there
are also weakened versions of IIA which imply \textit{dictatorship} for the
relevant preference aggregation rule. An obvious example concerning social
welfare functions is the restriction of IIA to subsets of alternatives of a
fixed cardinality (see Blau (1971)). More significantly, there are certain
consequences of IIA (weaker than IIA itself) that imply dictatorship of
preference aggregation rules when coupled with the Weak Pareto condition or
indeed \textit{any non-constancy constraint. }Notice that this fact holds
not only for social welfare functions, but also for preference aggregation
rules admitting \textit{any total binary relation }as their output (with no
transitivity or `consistency' requirement at all!). That is the case of
so-called \textit{Independent} \textit{Decisiveness} of aggregation rule $f$
requiring that any coalition which is able to enforce its strict preference
over a certain ordered pair of alternatives $(x,y)$ for \textit{some }%
preference profile (no matter what the preferences of others over $x,y$
happen to be) must also be \textit{decisive }for $(x,y)$: see Sen (1993),
Denicol\`{o} (1998), Quesada (2002). For another example of an Arrow-like
theorem for aggregation rules in the same vein (albeit in a much more
general setting) see Dani\"{e}ls, Pacuit (2008).}\textit{\ }and (ii) with
reference to \textit{strategy-proofness}, an issue that (as opposed to
agenda manipulation-proofness) was explicitly \textit{put aside} in the
original Arrowian analysis of social welfare functions (see Arrow (1963),
p.7).

Now, both of those tenets run counter to some views that are apparently
still widely held in the literature, and to which we now turn. To begin
with, the exceptional strength of IIA is sometimes downplayed or in any case
not fully appreciated. \textit{\ }One reason for that may be the (correct)
perception of the relationship of IIA to agenda-content
manipulation-proofness as combined with the (incorrect) view that
sequential-coupling of agenda formation to preference elicitation is the
only available possibility\footnote{%
Indeed, there is arguably no other way to make full sense of the following
statement from a well-known and highly respected scholar:\ `Independence of
Irrelevant Alternatives and therewith Binary Independence are eminently
reasonable assumptions to make in a realistic study of collective choice. I
know of no real-world collective-choice process that violates either
condition. Both formalize the idea that collective choices depend only on
such preferential data as could be revealed by \textit{voting.' }(Schwartz
(1986), 33). In that connection, similar comments apply to Bordes, Tideman
(1991).}. It is also possibly the case that IIA is occasionally confused
with its earlier counterpart named `Postulate of Relevancy' which is due to
Huntington (1938), and is explicitly quoted by Arrow himself as a source of
inspiration and `a condition analogous to' IIA\ (Arrow (1963), p. 27).
Notice, however, that while Huntington's `Postulate of Relevancy' may well
be quite similar in spirit to IIA, it is in fact \textit{much weaker }than
the latter\textit{\ }because it relies on a \textit{common language of
linearly ordered grades} (indeed, numbers\footnote{%
The examples considered by Huntington (1938) concern in fact competing teams
of equal size, and the relevant numbers/scores are uniquely determined by
the measurement of individual performances of each team's members. Thus, the
alternatives to be ranked are teams, while the agents are the \textit{shared
classifiers} for distinct members of each team (e.g. \textit{first }members
of some team, \textit{second} members of some team, and so on). Huntington
essentially contrasts team ranking by aggregation of members' ratings and
members' rankings, respectively, observing that the former method typically 
\textit{does satisfy} the `Postulate of Relevancy' while the latter \textit{%
does not}.}) to express \textit{absolute judgments} (as opposed to merely
comparative ones)\footnote{\textit{Majority judgment} as recently introduced
by Balinski and Laraki (Balinski, Laraki (2011)) denotes a family of
aggregation and voting mechanisms which typically satisfy the `Postulate of
Relevancy' while violating IIA (see also Vannucci (2019) for a detailed
discussion of strategy-proofness properties of majority judgment).}. A third
line of reasoning in support of IIA originates from a misleading
interpretation of a well-known theorem due to Satterthwaite (1975) that
establishes a tight connection between strategy-proof \textit{strict }social
choice functions and \textit{strict} social welfare functions that satisfy
IIA\footnote{%
Specifically, Satterthwaite's theorem establishes a one-to-one
correspondence between sovereign strategy-proof strict social choice
functions and sovereign strict social welfare functions that satisfy IIA and
the Weak Pareto principle, whenever the size of the set $A$ of alternatives
is not smaller than three (see Satterthwaite (1975), Theorem 2). A \textit{%
strict} social choice function (social welfare function, respectively) is a
social choice function (social welfare function, respectively) whose domain
(and codomain) is (are) appropriately restricted to \textit{linear orders}
of $A$.}. Indeed, Satterthwaite himself claims that such a theorem `creates
a strong new justification for [WP and] IIA as conditions that an ideal
social welfare function should satisfy' (Satterthwaite (1975), p. 207, with
some minor editing of mine)\footnote{%
Even a most respected, highly influential author such as Saari opts on his
part for the following concise rendition of the upshot of the
Satterthwaite's theorem mentioned in the text:\ `A nonmanipulable system
satisfies IIA' (Saari (2008), p.60).}. Notice, however, that the
`IIA-nonmanipulability' connection identified and discussed by Satterthwaite
concerns IIA as a property of a strict social welfare function and
nonmanipulability of the strict social choice function attached to the
former, and such nonmanipulability amounts to \textit{strategy-proofness }%
(and obviously \textit{not} agenda manipulation-proofness) of the latter. To
be sure, further interesting elaborations on such connections between IIA
and nonmanipulable aggregation rules are provided in Sato (2015).
Specifically, Sato considers \textit{four} notions of \textit{%
nonmanipulability }for strict social welfare functions in order to formulate
both agenda manipulation-proofness and strategy-proofness requirements,
respectively. Then, relying on the Kendall metric for linear orders, he
introduces a weak continuity condition for strict social welfare functions
called \textit{Bounded Response} \footnote{%
The Kendall distance $d_{K}$ between rankings is given by the minimal number
of transpositions of adjacents elements that is necessary to obtain one
linear order starting from another one.
\par
A strict social welfare function $f$ satisfies Bounded Response if $%
d_{K}(f(R_{N}),f(R_{N}^{\prime}))\leq1$ whenever two preference profiles $%
R_{N},R_{N}^{\prime}$ are the same except for the preference of a single
agent $i$, and $R_{i}$ and $R_{i}^{\prime}$ are \textit{adjacent} (i.e. $%
R_{i}^{\prime}$ is obtained from $R_{i}$ by permuting the $R_{i}$-ranks of 
\textit{a single} pair of alternatives with \textit{consecutive} $R_{i}$%
-ranks).}. The main result of Sato (2015) implies \textit{the equivalence}
of the following statements concerning a \textit{strict} social welfare
function $f$ for $(N,A)$: (1) $f$ satisfies Bounded Response and \textit{at
least one} of the four distinct agenda manipulation-proofness or
strategy-proofness conditions mentioned above; (2) $f$ satisfies Bounded
Response and \textit{each one }of the foregoing nonmanipulability
conditions; (3) $f$ satisfies Adjacency-restricted Monotonicity (AM)%
\footnote{%
The \textit{Adjacency-Restricted Monotonicity }condition for strict social
welfare functions simply requires that for any pair of `adjacent' profiles $%
R_{N},R_{N}^{\prime}$ and any $x,y\in A$, if [$yR_{i}x$, $xR_{i}^{\prime}y$
and $xf(R_{N})y$] then $xf(R_{N}^{\prime})y$ as well.} and the Arrowian 
\textit{IIA }condition. Thus, even factoring in AM (a very mild requirement
that is virtually undisputable) it turns out that IIA is in particular a
necessary condition of strategy-proofness \textit{only} for a \textit{%
specific }class of `\textit{weakly continuous'} and \textit{strict} social
welfare functions\footnote{%
It should also be emphasized that, when it comes to preference aggregation
problems, there is no reason to consider continuity conditions as
essentially `technical' and innocuous. It is indeed well-known that
anonymity and idempotence of a social welfare function (or indeed of
virtually any preference aggregation rule for arbitrary profiles of total
preorders) are inconsistent with preservation of `preference proximity' (see
Baigent (1987)). See also Lauwers, Van Liedekerke (1995), Huang (2004) and
Saari (2008) for more general considerations on the difficulties raised by
continuity properties for aggregation rules. Clearly enough, requiring
proper `responsiveness' of a preference aggregation rule is one thing, and
insisting on its `continuity' quite another.}. In short, a closer inspection
of both Satterthwaite (1975) and Sato (2015) confirms that the most
interesting results they contribute are in fact \textit{silent} on necessary
and/or sufficient conditions for strategy-proofness of \textit{general}, 
\textit{unrestricted social welfare functions}.

All of the above suggests that both agenda manipulation-proofness and
strategy-proofness of a proper consensus-based social welfare function do
indeed require that IIA be \textit{either just dropped or at the very least
considerably relaxed}.

The independence condition used in the present paper, namely $M_{\mathcal{X}%
} $-Independence, can be indeed regarded as a drastic relaxation of IIA when
applied to social welfare functions. It was first introduced by Monjardet
(1990) and explicitly related to IIA and Arrowian aggregation models, but
not at all to strategy-proofness issues (or, for that matter, to agenda
manipulation-proofness issues) \footnote{%
It should be noted, however, that conditions strictly related to M$_{%
\mathcal{X}}$-Independence are deployed in Dietrich, List (2007b) to study
strategy-proofness properties in judgment aggregation as discussed below in
the present section.} .

Unsurprisingly, several alternative weakenings of IIA have been proposed in
the earlier literature. An entire set of substantially relaxed versions of
IIA was first introduced and discussed by Hansson (1973) with no reference
whatsoever to nonmanipulability issues of any sort\footnote{%
Afriat (1987) -first appeared as a 1973 conference paper- is also to be
credited for an early criticism of IIA. That criticism is mainly motivated
by the understanding of an Arrowian social welfare function as a way of
modeling the process of `voting for an order' and, again, with no explicit
reference whatsoever to agenda manipulation issues.}. The strongest of them
(i.e. the \textit{least} dramatic relaxation of IIA, denoted by Hansson as
Strong Positionalist Independence (SPI)\footnote{%
The label comes from the fact that SPI is of course satisfied by
`positionalist' or score-based aggregation rules including the Borda Count
rule (which assigns to every alternative $x$ a score given by the sum of its
individual ranks, defined as the sizes of the sets of alternatives which are
classified as strictly worse than $x$ itself).
\par
{}}) requires invariance of aggregate preference between any two
alternatives $x,y$ for any pair of preference profiles such that their 
\textit{restrictions to} $\left \{ x,y\right \} $\textit{\ are identical,
and for every agent/voter the supports of the respective closed preference
intervals having }$x$ \textit{and }$y$\textit{\ as their extrema are also
identical. }Incidentally\textit{, }SPI has been recently rediscovered,
relabeled as Modified IIA, and provided with a new motivation by Maskin
(2020). Indeed, Maskin points out that SPI enforces resistance of the
relevant aggregation rule to certain sorts of \ `vote splitting' effects,
thereby connecting SPI to manipulation issues, including strategic
manipulation. Notice, however, that Maskin's proposal is aimed at
strategy-proofness of the `\textit{maximizing' social choice function }%
induced by a certain social welfare function (as opposed to
strategy-proofness of the social welfare function itself). In a similar
vein, another weakening of IIA that is even stronger than SPI has been
proposed by Saari under the label `Intensity form of IIA' (IIIA). IIIA
requires invariance of aggregate preference between any two alternatives $%
x,y $ for any pair of preference profiles such that\textit{\ for every
agent/voter the rank (or score) difference between }$x$ \textit{and }$y$%
\textit{\ is left unchanged }from one profile to the other\textit{\ }(see
Saari (1995) and (1998))\footnote{%
Arguably, Saari's IIIA can also be regarded as a formalization of the
criticism of IIA originally advanced by Dahl (1956) with his advocacy of
aggregation rules based on intensity of individual preferences. Notice that
IIIA is indeed satisfied by some positional aggregation rules such as the
Borda Count but also by majority judgment as discussed above.}.

A further weakening of IIA in a quite different vein is due to Huang (2014),
under the label \textit{Weak Arrow's Independence }(WIIA). In plain words, a
social welfare function $f$ satisfies WIIA if, for any pair $R_{N}$,$%
R_{N^{\prime }\text{ }}$ of profiles of total preorders and any pair $x$,$y$
of alternatives such that the preferences between $x$ and $y$ of every agent 
$i$ in $N$ are the same in $R_{N}$ and $R_{N^{\prime }}$, the following
condition holds: if $x$ is \textit{strictly} \textit{preferred} to $y$
according to social preference $f(R_{N})$ then $x$ is \textit{preferred }%
(i.e. either strictly preferred or indifferent) to $y$ according to social
preference $f(R_{N}^{\prime })$. Notice the main difference between WIIA and
virtually all of the other weakenings of IIA considered in the present work:
while the other weakenings \textit{strenghten the hypothetical clause} of
IIA and leave its consequent unaltered, WIIA keeps the hypothetical clause
of IIA unaltered and \textit{weakens its consequent}. It should also be
emphasized that the overt motivation of WIIA is just finding an escape route
from the strictures of Arrow's theorem \textit{without any explicit
consideration of agenda manipulation-proofness or strategy-proofness
properties}. On the other hand, it is worth mentioning that Huang's positive
proposal for social welfare functions (what he denotes as `Weak Arrow's
Framework') consists in replacing IIA and WP with WIIA and BP (the Basic
Pareto principle -or Weak Pareto Condition in Huang's own terminology- which
is also used in our Proposition 3), and allowing for social indifference
classes that contain both a strictly Pareto dominated alternative and some
of its strict Pareto improvements (such classes are precisely the \textit{%
stalemates }previously considered in the present work, that are denoted as 
\textit{singularities }by Huang).

Remarkably, even at a first glance one conspicuous difference between $M_{%
\mathcal{X}}$-Independence and SPI (or IIIA and WIIA) stands out
immediately: the former relies heavily on the structure of the outcome set,
while SPI, IIIA and WIIA\ only impinge upon the relevant preference
profiles, completely disregarding any specific feature/structure of the
relevant outcome set (namely the set of all total preorders of the set $A$
of basic alternatives).

This crucial difference and its significant import can be further clarified
and fully appreciated by reconsidering all the relaxations of IIA mentioned
above from the common perspective of `\textit{aggregation by binary issues}'
that encompasses them all.

The \textit{binary aggregation model }originates with Wilson (1975) and has
been further extended by Rubinstein, Fishburn (1986)\footnote{%
To be sure, the original work by Wilson only considers the \textit{finite }%
case, but Wilson's framework can also be extended to an infinite number of
issues, and to non-binary issues (see Dokow, Holzman (2010c)). Indeed, one
such extension is covered in Rubinstein, Fishburn (1986). Since the present
paper is only concerned with finite social welfare functions, however, we
shall only consider the basic binary aggregation model with a \textit{finite}%
$\mathbb{\ }$number of issues.}: a finite number of $k$ issues are
considered for a collective yes/no judgment (the output) to be based on some
profile of individual yes/no judgments on each issue (the input), under some 
\textit{feasibility constraints} (usually the same, but possibly different)
imposed, respectively, on inputs and outputs\footnote{%
In more recent contributions coming from the computational social choice and
artificial intelligence research communities `\textit{integrity constraints' 
}is the most commonly used label to denote such constraints (see e.g.
Grandi, Endriss (2013)).\textit{\ }}. Thus, the basic aggregation rules for
a set $N$ of agents and a set $K=\left \{ 1,...,k\right \} $ of binary
issues are given by functions $f:X^{N}\rightarrow X$ with $X\subseteq \left
\{ 0,1\right \} ^{K}$. This model has also be shown to be equivalent to the
basic model of \textit{judgment aggregation }where the judgments to be
aggregated amount to acceptance/rejection of every element of an \textit{%
agenda} of interconnected formulas of a suitable formal language
representing propositions (see Dokow, Holzman (2009)). Indeed, several
versions of Arrow's general (im)possibility theorem for social welfare
functions have been explicitly shown to follow as a special interesting case
under both the \textit{feasible binary aggregation }and the \textit{judgment
aggregation }frameworks (see e.g. Dokow, Holzman (2010a, 2010b) for the
former and Dietrich, List (2007a), Van Hees (2007), Mongin (2008), Dani\"{e}%
ls, Pacuit (2008), Porello (2010) for the latter)\footnote{%
In particular, Mongin (2008) introduces a specific \textit{weakening of IIA }%
for the standard judgment aggregation model, by restricting the scope of IIA
to atomic propositional formulas, and still obtains a version of Arrow's
(im)possibility theorem under WP. As previously mentioned, Dani\"{e}ls,
Pacuit (2008) offers another characterization of dictatorial rules in a
quite general judgment aggregation framework using just some consequences of
IIA as combined with non-constancy and neutrality conditions. Furthermore,
it has been shown that Arrowian characterizations of dictatorial aggregation
rules by IIA and idempotence hold for other disparate domains including
arbitrary single-valued choice functions on finite sets (Shelah (2005)) and
task assignments (Dokow, Holzman (2010c)).}.\textit{\ }

An additional and most convenient perspective for the finite version of the
binary aggregation model of our concern here is provided by some joint work
of Nehring and Puppe (see in particular Nehring, Puppe (2007),(2010)). To be
sure, Nehring, Puppe (2007) is mainly concerned with \textit{strategy-proof
social choice functions} as defined on profiles of total preorders on finite
sets. Conversely, Nehring, Puppe (2010) is focussed on an `abstract' class
of Arrowian aggregation problems including preference aggregation and, more
specifically, social welfare functions, but it does \textit{not }address
issues concerning their strategy-proofness properties. However, social
choice functions with the top-only property\textit{\footnote{%
A social choice function $f:(\mathcal{R}_{A})^{N}\rightarrow A$ satisfies
the \textit{top-only property }if $f(R_{N})=f(R_{N}^{\prime})$ whenever $%
t(R_{i})=t(R_{i}^{\prime})$ for each $i\in N$, and $|t(R_{i})|=|t(R_{i}^{%
\prime})|=1$ for all $i\in N$ (with $t(R_{i}):=\left \{ x\in A:xR_{i}y\text{
for all }y\in A\right \} $).}} may be regarded as aggregation rules endowed
with a specific domain of total preorders, and the class of Arrowian
aggregation rules considered in Nehring, Puppe (2010) does include the case
of preference aggregation rules in finite median semilattices. Specifically,
Nehring and Puppe attach to any \textit{finite }outcome space $X$ a certain
finite hypergraph $\mathbb{H}=(X,\mathcal{H})$ denoted as \textit{property
space, }where the set $\mathcal{H}\subseteq \mathcal{P}(X)\smallsetminus
\left \{ \varnothing \right \} $ of (nonempty) hyperedges or \textit{%
properties }of outcomes/states in $X$\textit{\ }is \textit{%
complementation-closed }and \textit{separating }(namely $H^{c}:=X\setminus
H\in \mathcal{H}$ whenever $H\in \mathcal{H}$, and for every two \textit{%
distinct }$x,y\in X$ there exists $H_{x^{+}y^{-}}\in \mathcal{H}$ such that $%
x\in$ $H_{x^{+}y^{-}}$and $y\notin H_{x^{+}y^{-}}$). Such a property space $%
\mathbb{H}$ models the set of all \textit{binary} properties of outcomes
that are regarded as relevant for the decision problem at hand. Thus, 
\textit{binary issues are modeled here as pairs }$(H,H^{c})$\textit{\ of
complementary properties }and, as it is easily checked, \textit{both the
feasible binary aggregation and the judgment aggregation models can be
immediately reformulated as aggregation models in property spaces.} Then, a
betweenness relation $B_{\mathbb{H}}\subseteq X^{3}$ is introduced by
stipulating that $B_{\mathbb{H}}(x,y,z)$ holds precisely when $y$ satisfies
all the properties shared by $x$ and $z\ $\footnote{%
In particular, a nonempty subset $Y\subseteq X$ is said to be \textit{convex 
} for $\mathbb{H=}(X,\mathcal{H})$ if for every $x,y\in Y$ and $z\in X$, if $%
B_{\mathbb{H}}(x,z,y)$ then $z\in Y$, and \textit{prime }(or a \textit{%
halfspace}) for $\mathbb{H}$ if both $Y$ and $X\setminus Y$ are convex for $%
\mathbb{H}$ and $\left \{ Y,X\setminus Y\right \} \subseteq \mathcal{H}$.}.
Moreover, \textit{single-peaked }preference domains on $X$ can be defined
relying on $B_{\mathbb{H}}$. In particular, $B_{\mathbb{H}}$ is said to be 
\textit{median }if for every $x,y,z\in X$ there exists a unique $m_{xyz}\in
X $ such that $B_{\mathbb{H}}(x,m_{xyz},y)$, $B_{\mathbb{H}}(x,m_{xyz},z)$,
and $B_{\mathbb{H}}(y,m_{xyz},z)$ hold\footnote{%
In that case, $\mathbb{H}$ is said to be a \textit{median property space, }$%
(X,m^{\mathbb{H}})$ (where $m^{\mathbb{H}}:X^{3}\rightarrow X$ is defined by
the rule $m^{\mathbb{H}}(x,y,z)=m_{xyz}$ for every $x,y,z\in X$) is a 
\textit{median algebra}, and for each $u\in X$ the pair $(X,\vee_{u})$
(where $x\vee_{u}y=y$ iff $m^{\mathbb{H}}(x,y,u)$ $=y$ for some $u\in X$) is
a \textit{median join-semilattice} having $u$ as its maximum.}. Of course, a
main advantage of that second-order representation of binary issues is the
possibility to \textit{focus on several different property spaces which are
defined on the very same ground} \textit{set of alternative states. }

The following key results are obtained by Nehring and Puppe: (i) the class
of all idempotent social choice functions which are strategy-proof on the
domain of single-peaked preferences thus defined are characterized in terms
of voting by binary issues through a certain combinatorial property\footnote{%
The combinatorial property mentioned in the text is the so-called `\textit{%
Intersection Property}' which requires that for every minimal inconsistent
set of properties, it must be the case that any selection of winning
coalitions for the corresponding binary issues has a non-empty intersection.}
of the families of winning coalitions for the relevant issues and (ii) 
\textit{if the property space is median }then\textit{\ }such combinatorial
property is definitely met, and consequently non-dictatorial neutral and/or
anonymous strategy-proofs aggregation rules including the simple majority
rule are available (Nehring, Puppe (2007), Theorems 3 and 4). Furthermore,
in Nehring, Puppe (2010) the very same theoretical framework is deployed to
analyze preference aggregation and social welfare functions. In particular,
several `classical' properties for social welfare conditions including the
Arrowian Independence of Irrelevant Alternatives (IIA) property can be
reformulated in \textit{more general terms} which depend on the
specification of the relevant property space\footnote{%
Specifically, given a property space $\mathbb{H}=(\mathcal{R}_{A},\mathcal{H}%
)$, such a generalized $IIA$ for a social welfare function $f:(\mathcal{R}%
_{A})^{N}\longrightarrow \mathcal{R}_{A}$ can be defined as follows: for
every $H\in \mathcal{H}$ and $R_{N},R_{N}^{\prime}\in(\mathcal{R}_{A})^{N}$
such that $\left \{ i\in N:R_{i}\in H\right \} =\left \{ i\in
N:R_{i}^{\prime}\in H\right \} $,
\par
if $f(R_{N})\in H$ then $f(R_{N}^{^{\prime}})\in H$ as well. Of course the
original Arrowian version of such a generalized IIA may be obtained by
taking $\mathcal{H}:=\left \{ H_{(x,y)}:x,y\in A\right \} $ with $%
H_{(x,y)}:=\left \{ R\in \mathcal{R}_{A}:xRy\right \} $. That is so because
it is well-known that in the Arrowian aggregation framework IIA is
equivalent to its \textit{binary} version (namely, its restriction to
arbitrary pairs of alternatives).}: it follows that several versions of such
a \textit{generalized IIA condition }can be considered. But then, as it
turns out, (iii) the \textit{versions of generalized IIA attached to median
property spaces} are consistent with anonymous and (weakly) neutral social
welfare functions including those induced by majority-based aggregation
rules (Nehring, Puppe (2010), Theorem 4). Interestingly, a specific example
of a median property space for the set of all total preorders is also
provided by Nehring and Puppe, namely the one whose issues consist in asking
for each non-empty $Y\subseteq X$ and any total preorder $R$ whether or not $%
Y$ is a \textit{lower contour }of $R$ with respect to some outcome $x\in X$.%
\footnote{%
Thus, the property space suggested here is $\mathbb{H}^{\circ}:=\left \{ 
\mathcal{R}_{A},\mathcal{H}^{\circ }\right \} $, where
\par
$\mathcal{H}^{\circ}:=\left \{ H_{L}:\varnothing \neq L\subseteq A\right \} $
and
\par
$H_{L}:=\left \{ 
\begin{array}{c}
R\in \mathcal{R}_{A}:\text{for some }x\in A \\ 
L=\left \{ y\in A:xRy\right \}%
\end{array}
\right \} $.} By contrast, it can be easily checked that when translated
into the property-space framework SPI and IIIA correspond to \textit{%
non-median property spaces.}\footnote{%
Indeed, the most natural property-space attached to SPI is
\par
$\mathcal{H}_{SPI}:=\left \{ 
\begin{array}{c}
H_{(x,y,B)}:x,y\in A \\ 
B\subseteq A\setminus \left \{ x,y\right \}%
\end{array}
\right \} $ with
\par
$H_{(x,y,B)}:=\left \{ 
\begin{array}{c}
R\in \mathcal{R}_{A}:\text{ }\left \{ a\right \} \times \text{ }B\subseteq R%
\text{ and } \\ 
B\times \left \{ b\right \} \subseteq R \\ 
\text{if }\left \{ a,b\right \} =\left \{ x,y\right \}%
\end{array}
\right \} $.
\par
Similarly, the most natural property-space attached to IIIA is
\par
$\mathcal{H}_{IIIA}:=\left \{ 
\begin{array}{c}
H_{(x,y,k)}:x,y\in A, \\ 
k\leq|A|-2%
\end{array}
\right \} $ with
\par
$H_{(x,y,k)}:=\left \{ 
\begin{array}{c}
R\in \mathcal{R}_{A}:\text{ either }I_{x,y}=\left \{ z\in
A:xP(R)zP(R)y\right \} \text{ } \\ 
\text{and }k=|I_{x,y}| \\ 
\text{or }I_{x,y}=\left \{ z\in A:yP(R)zP(R)x\right \} \text{ } \\ 
\text{and }|I_{x,y}|=k%
\end{array}
\right \} $.
\par
{}It can be shown that $\mathcal{H}_{SPI\text{ }}$and $\mathcal{H}_{IIIA%
\text{ }}$are \textit{not} median property spaces since both of them contain
minimal inconsistent subsets of properties of size three. To check validity
of that statement, just consider any triplet $\left \{
H_{x,y,\emptyset},H_{y,z,\emptyset},H_{z,x,\emptyset}\right \} \subseteq%
\mathcal{H}_{SPI\text{ }}$%
\par
and $\left \{ H_{x,y,0},H_{y,z,0},H_{z,x,0}\right \} \subseteq \mathcal{H}%
_{IIIA\text{ }}$%
\par
with $x\neq y\neq z\neq x$.}

The overlappings between such results and those presented here are
remarkable, along with some sharp differences which make them mutually
independent. Since any finite median semilattice is indeed an example of a
finite median algebra\footnote{%
Specifically, a finite median join-semilattice can be regarded as a generic
instance of a finite `\textit{pointed'} median algebra, having one of its
elements singled out (that point corresponds to the top element of the
semilattice).}, and is consequently representable as a median property space%
\footnote{%
For instance, it is \textit{always} possible to represent a (finite) median
algebra as a (finite) property space by taking as properties its \textit{%
prime} sets as defined through its median betweenness (see e.g. Bandelt, Hedl%
\'{\i}kov\'{a} (1983), Theorem 1.5, and footnote 52 above for a definition
of prime sets). It is important to observe that in general a finite median
algebra or ternary space admits of several representations by distinct
median property spaces. By contrast, a ternary (finite) algebra or space
which is not median can only be represented by (finite) property spaces
which are \textit{not} median.}, all of the Nehring and Puppe's results
mentioned above \textit{do apply }to finite median semilattices as a special
case. Notice however that our results provide a characterization of (finite) 
\textit{strategy-proof} social welfare functions which is also both \textit{%
more explicit} (it includes a polynomial description of some such rules) and 
\textit{more comprehensive }(it is a complete characterization in that it is
not limited to sovereign and idempotent ones). Moreover, our treatment of
social welfare functions can also be translated in terms of a \textit{median}
property space, but a \textit{different one} from that considered by Nehring
and Puppe. In fact, in our case the set of relevant properties correspond to
the meet-irreducibles of the semilattice of total preorders, namely the
total preorders having just \textit{two} indifference classes, or
equivalently the binary ordered classifications of basic alternatives as 
\textit{good }or \textit{bad}, respectively. Accordingly, the collection of
relevant issues consist in asking, for each binary good/bad classification
of basic alternatives and any total preorder $R$, \textit{whether the latter
is consistent with the given binary classification or not}\footnote{%
Thus, the appropriate version of generalized IIA in our own model is $%
\mathbb{H}^{\ast }:=(\mathcal{R}_{A},\mathcal{H}^{\ast })$ with $\mathcal{H}%
^{\ast }:=\left \{ 
\begin{array}{c}
H_{A_{1}A_{2}}:A_{1}\neq \varnothing \neq A_{2} \\ 
A_{1}\cap A_{2}=\varnothing \text{, }A_{1}\cup A_{2}=A%
\end{array}%
\right \} $%
\par
$H_{A_{1},A_{2}}:=\left \{ R\in \mathcal{R}_{A}:R\subseteq
R_{A_{1}A_{2}}\right \} $%
\par
and $R_{A_{1}A_{2}}$ is of course the two-indifference-class total preorder
having $A_{1}$ and $A_{2}$ as top and bottom indifference classes,
respectively. Notice that both $\mathbb{H}^{\ast }$ and Nehring-Puppe's $%
\mathbb{H}^{\circ }$as previously defined (see footnote 56\ above) are 
\textit{median }property spaces, while the original Arrowian $\mathbb{H}$ is
not.
\par
{}}. Thus, proper consensus-based social welfare functions that are agenda
manipulation-proof and even strategy-proof can be defined by \textit{binary
aggregation}, provided that the set of relevant binary issues is carefully
selected, and in fact \textit{expanded }if the basic alternatives are\textit{%
\ more than three}: notice that, when expressed in terms of properties of $%
\mathcal{R}_{A}$ with $|A|=m$, the size of the set of actually relevant
binary issues (for each individual preference relation of any profile) is $%
m(m-1)$ for IIA, and $2(2^{m-1}-1)$ for $M_{\mathcal{X}}$-Independence.%
\textit{\ }This point is strongly consonant with one of the main arguments
in Saari (2008), lamenting the enormous loss of information enforced by the
Arrowian IIA. It also amounts to a special instance of a recurrent theme in
the social choice-theoretic literature, namely emphasizing the link between
Arrow's theorem and the strictures of the preference-information base
enforced by IIA and the other Arrowian axioms (see e.g. the classic Sen
(2017) for extensive elaborations on that topic). Notice, however, that
while changes and/or enrichments of the Arrowian input-format figure
prominently among the invoked remedies for the aforementioned strictures,
the relaxations of IIA we have been considering stick to a \textit{fixed,
standard input-format consisting of profiles of total preorders}.

\end{document}